\documentclass[11pt]{article} 
\usepackage{ifthen}
\usepackage{mdwlist}
\usepackage{amsmath,amssymb,amsfonts,amsthm}
\usepackage{bm}
\usepackage[colorlinks,citecolor=blue]{hyperref}
\usepackage{enumitem}
\usepackage{graphicx}
\usepackage{xspace}
\usepackage{verbatim}
\usepackage{algorithm}
\usepackage{algorithmic}
\usepackage[margin=1in]{geometry}
\usepackage[dvipsnames]{xcolor}
\usepackage{thmtools}
\usepackage{thm-restate}
\usepackage{mathtools}
\usepackage{multicol}
\usepackage{multirow}
\usepackage[T1]{fontenc}
\usepackage{latexsym}
\usepackage{epsfig}
\usepackage{epstopdf}
\usepackage{subcaption}
\usepackage{bbm}
\usepackage{tikz}
\usepackage{todonotes}
\usepackage{lmodern}
\usepackage{xcolor}
\usepackage{float}

\newcommand{\Upper}[1]{\mathsf{maxlen}(#1)}
\renewcommand{\Pr}{\mathbf{Pr}}
\newcommand{\rank}{\mathsf{rank}}

\def\vertex{\mathsf{node}}

\def\pred{\mathsf{pred}}
\def\succ{\mathsf{succ}}
\def\leftt{\mathsf{left}}
\def\rightt{\mathsf{right}}
\def\interior{\mathsf{interior}_H} 
\def\leng{\mathsf{len}}
\def\eps{\epsilon}
\def\calX{\mathcal{X}}
\def\calA{\mathcal{A}}

\newcommand{\Radius}[1]{\mathsf{radius}_H(#1)}
\newcommand{\leftradius}[1]{\mathsf{L}\text{-}\mathsf{radius}_H(#1)}
\newcommand{\rightradius}[1]{\mathsf{R}\text{-}\mathsf{radius}_H(#1)}
\newcommand{\Appearances}{\mathsf{endpoints}}

\newcommand{\restrict}[2]{\ensuremath{{#1}\hspace{-0.1cm}\downharpoonright_{#2}}}

\newcommand{\Cover}[1]{\mathsf{cover}_{#1}}
\newcommand{\Covereven}[1]{\mathsf{cover}_{#1}^{\mathsf{even}}}
\newcommand{\Coverodd}[1]{\mathsf{cover}_{#1}^{\mathsf{odd}}}
\newcommand{\Coverboundary}[1]{\mathsf{cover}_{#1}^{\mathsf{boundary}}}

\newcommand{\arcs}{\ensuremath{\mathcal{A}}}

\newcommand{\improvement}[2]{\mathsf{impvm}_{#1}(#2)}

\newcommand{\longradius}{\ensuremath{\mathsf{long}}-\ensuremath{\mathsf{radius}}~}
\newcommand{\shortradius}{\ensuremath{\mathsf{short}}-\ensuremath{\mathsf{radius}}~}

\newcommand{\good}[1]{\mathsf{good}_H(#1)}
\newcommand{\bad}[1]{\mathsf{bad}_H(#1)}
\newcommand{\dualbad}[1]{\mathsf{dual}\text{-}\mathsf{bad}_H(#1)}

\newcommand{\Roglin}{R$\ddot{\text{o}}$glin~}

\newcommand{\Touches}{\mathsf{Overlap}}
\newcommand{\NonTouches}{\mathsf{NonOverlap}}

\newcommand{\Xichen}[1]{\textcolor{Purple}{(Xi: #1)}}

\newcommand{\Xinzhi}[1]{\textcolor{NavyBlue}{(Xinzhi: #1)}}
\newcommand{\Manolis}[1]{\textcolor{Blue}{(Manolis: #1)}}

\newcommand{\answer}[1]{#1}

\newcommand{\group}{\mathcal{D}}
\newcommand{\chunk}{\mathcal{C}}
\usepackage{lipsum}  
\usepackage{thmtools}
\usepackage{thm-restate}

\newtheorem{theorem}{Theorem}[section]
\newtheorem{fact}{Fact}[section]
\newtheorem{lemma}{Lemma}[section]

\newtheorem{corollary}{Corollary}[section]
\newtheorem{definition}{Definition}[section]

\usetikzlibrary{positioning, shapes.geometric}

\title{Smoothed complexity of local Max-Cut and binary Max-CSP\vspace{0.2cm}}
\author {
Xi Chen\thanks{Supported by NSF IIS-1838154 and NSF CCF-1703925}\\
Columbia University\\
\tt{xichen@cs.columbia.edu}
\and
Chenghao Guo\thanks{Part of this work was done while the author was visiting Columbia University.} \\
IIIS, Tsinghua University\\
\tt{guoch16@mails.tsinghua.edu.cn}
\and
Emmanouil V. Vlatakis-Gkaragkounis\thanks{Supported by NSF CCF-1703925, NSF CCF-1763970, NSF CCF-1814873 and NSF CCF-1563155.}\\
Columbia University\\
\tt{emvlatakis@cs.columbia.edu}
\and
Mihalis Yannakakis\thanks{Supported by NSF CCF-1703925 and CCF-1763970.} \\
Columbia University\\
\tt{mihalis@cs.columbia.edu}
\and
Xinzhi Zhang\thanks{Part of this work was done while the author was visiting Columbia University.} \\
IIIS, Tsinghua University\\
\tt{zhang-xz16@mails.tsinghua.edu.cn}
}

\begin{document}

\maketitle
\begin{abstract}
We show that the smoothed complexity of the FLIP algorithm for local Max-Cut is at most $\smash{\phi n^{O(\sqrt{\log n})}}$, where $n$ is the number of nodes in the graph and $\phi$ is a parameter that measures the magnitude of perturbations applied on its edge weights. 
This improves the previously best upper bound
  of $\phi n^{O(\log n)}$ by Etscheid and \Roglin \cite{etscheid2015smoothed}.
Our result is based on an analysis of long sequences   of flips, which shows~that~it is very unlikely for every flip in a long sequence to incur a positive but small improvement in the cut weight. 
We also extend the same upper bound on the smoothed complexity of FLIP 
to all binary Maximum Constraint Satisfaction Problems.
\end{abstract}

\clearpage

\section{Introduction}

Local search is one of the most prominent algorithm design paradigms for combinatorial optimiza\-tion problems.
A local search algorithm begins with an initial candidate
  solution and then follows a path by iteratively moving to a better neighboring solution  
  until a local optimum is reached. 
Many algorithms currently deployed in practice  
  are based on local search, and all the empirical evidence suggests that they typically perform very well in practice,
  rarely running into long paths before reaching a local optimum.
  
However, despite their wide success in practice, 
the performance of many local search algorithms lacks
  rigorous justifications. 
A recurring phenomenon is that a local search algorithm is usually efficient in practice but
  analysis under the worst-case framework indicates the opposite --- that
  the algorithm has exponential running time due to
  delicate pathological instances that one may never encounter in practice. 
A concrete (and probably one of the simplest) example of this phenomenon is the \emph{FLIP} algorithm for the \emph{local Max-Cut} problem.
 
Given an undirected graph $G=(V,E)$ with edge weights $(X_e: e\in E)$ (wlog in $[-1,1]$),
  the local Max-Cut problem is to find a partition of $V$ into two sets $V_1$ and $V_2$
  such that the weight of the corresponding cut 
  (the sum of weights of edges with one node in $V_1$ and the other in $V_2$)
  cannot be improved by moving one of the nodes to the other set.
To find a local max-cut, the FLIP algorithm starts with an initial partition and keeps moving nodes to the other side, one by one, as long as the move increases the weight of the cut,
  until no local improvement can be made. 
  Note that the FLIP algorithm, similar to the simplex algorithm, is really a family of algorithms since one 
  can apply different rules, deterministic or randomized, to pick the next node when more than one nodes can improve the cut.
The local Max-Cut problem is known to be PLS-complete \cite{schaffer1991simple}, where PLS is a complexity class introduced by \cite{johnson1988easy} to 
  characterize  local search problems.
A consequence of the proof of the completeness result is that FLIP takes exponential time to solve local Max-Cut in the worst case, regardless of the pivoting rule used \cite{schaffer1991simple}.
The local Max-Cut problem can be viewed equivalently as the problem of finding a pure Nash equilibrium in a \emph{party affiliation game} \cite{fabrikant2004complexity}.
In this case, the FLIP algorithm corresponds to the better response dynamics for the game.
The local Max-Cut problem is also closely related to the problem of finding
a stable configuration in a neural network in the
  Hopfield model \cite{hopfield} (see Section \ref{sec:maxcsp} for the definition). In this case the FLIP algorithm corresponds to the natural asynchronous dynamics where in each step an unstable node flips its state, and the process repeats until the network converges to a stable configuration.

Max-Cut is an example of a Maximum Binary Constraint Satisfaction Problem (Max-2CSP).
In a general \emph{Max-2CSP}, 
  the input consists of a set of Boolean variables and a set of constraints with weights over some pairs of variables. The problem is then to find an assignment to the variables that maximizes the sum of weights of satisfied constraints.
So Max-Cut is the special case when~all constraints are XOR of the two variables.
Other well-studied special cases include Max-2SAT (Maximum Satisfiability when every clause has at most two literals),
 and Max-Directed Cut (the max-cut problem for weighted directed graphs); see Section \ref{sec:maxcsp} for their definitions.
We can consider more generally the \emph{Binary Function Optimization Problem} (or BFOP in short),
  where instead of constraints we have functions over some pairs of variables and the objective function
  is a weighted sum of these functions (again see Section \ref{sec:maxcsp} for the formal definition).
The FLIP algorithm can be used to find local optima for general Max-2CSP and BFOP, 
  where flipping the assignment of any single variable cannot improve the objective function.
  
In this paper we study the \emph{smoothed complexity} of the FLIP algorithm for local Max-Cut,
   Max-2CSP and BFOP.
The smoothed analysis framework was introduced by 
  Spielman and Teng \cite{spielman2004smoothed} to 
  provide rigorous
  justifications for the observed good practical performance of the simplex
  algorithm (the standard local search algorithm for Linear Programming), even though the simplex algorithm is known to take exponential-time in the worst case for most common pivoting rules (e.g. \cite{klee1970good}).
Since then, smoothed analysis has been applied in a range of areas; see \cite{spielman2009smoothed}.
Specialized to the local Max-Cut problem, the edge
weights of the given undirected graph $G=(V,E)$ are assumed to be drawn independently from a 
   vector $\calX=(\calX_e: e\in E)$ of probability distributions, one for each edge.
Each $\calX_e$ is a distribution supported on $[-1,1]$ and its density function is bounded from above
  by a parameter $\phi>0$. 
 {\color{blue} }{Notice that as $\phi\to 1/2$, the model approaches the average-case analysis framework for uniform edge weights. 
A related alternative model for smoothed analysis is to allow an adversary to pick arbitrary weights $w_e$, which are then perturbed by adding a small random perturbation $\mathcal{Z}_e$,  i.e. the edge weights are $\calX_e=w_e+\mathcal{Z}_e$.
In this case, $\phi$ corresponds to the maximum value of the pdf of $\mathcal{Z}_e$.}

The question is to give an upper bound $T(n,\phi)$ such that for any $G$ and $\calX$,
  the FLIP algorithm terminates within $T(n,\phi)$ steps with high probability (say $1-o_n(1)$)
  over the draw of edge weights $X\sim \calX$ 
  (where we use $X\sim \calX$ to denote independent draws of $X_e\sim \calX_e$).
  
The best result for $T(n,\phi)$ before our work is the quasipolynomial upper bound $\phi n^{O(\log n)}$ 
  by Etscheid and \Roglin \cite{etscheid2015smoothed}, based on a rank-based approach which we review in Section \ref{sec:rank}.
Before their work, polynomial upper bounds were obtained by Els$\ddot{\text{a}}$sser and Tscheuschner \cite{elsasser2011settling}
  and by Etscheid and \Roglin \cite{etscheid2015squaredeclidean} for special cases either when $G$ has $O(\log n)$ degree or 
  when $G$ is a complete graph with edge weights given by Euclidean distances.
After the work of \cite{etscheid2015smoothed}, Angel et. al \cite{angel2017local} obtained a polynomial upper bound for $T(n,\phi)$ when $G$ is a complete graph. 
Their polynomial bound was further improved by Bibak et. al \cite{bibak2019improving}, again for complete graphs. 
\subsection{Our results}

We prove a $\phi n^{O(\sqrt{\log n})}$ upper bound for the 
  smoothed complexity of FLIP for local Max-Cut:

\begin{restatable}[]{theorem}{maintheorem}
\label{maintheorem}
Let $G=(V,E)$ be an undirected graph over $n$ vertices, and let
  $\calX=(\calX_e:e\in E)$~be a sequence of probability distributions 
  such that every $\calX_e$ is supported on $[-1,1]$ and has its density function bounded from above by
  a parameter $\phi>0$.
Then with probability at least $1-o_n(1)$ over~the draw of edge weights $X\sim \calX$,
  any implementation of the FLIP algorithm running on $G$ and~$X$ takes at most $\phi n^{O(\sqrt{\log n})}$ many steps to terminate.
\end{restatable}

Given $G$ and edge weights $X$, we define the (directed) configuration graph they form as follows:
  vertices of the graph correspond to configurations (or partitions) $\gamma:V\rightarrow \{-1,1\}$; 
  there is an edge from $\gamma$ to $\gamma'$ if $\gamma'$ can be obtained from $\gamma$ by moving 
  one node
  and the weight of $\gamma'$ is strictly larger than that of $\gamma$ under $X$, i.e., each edge is 
  a move that strictly improves the cut weight.
Theorem \ref{maintheorem} is established by showing that, 
  with probability at least $1-o_n(1)$ over~the draw of $X\sim \calX$,
  there is no directed path longer than $\phi n^{O(\sqrt{\log n})}$ in the configuration graph formed by $G$ and $X$.



We also extend Theorem \ref{maintheorem} to obtain the same upper bound for the smoothed complexity of the FLIP algorithm running 
  on Max-2CSP and BFOP.

\begin{theorem}\label{maxcsptheorem}
Let $I$ be an arbitrary instance of a Max-2CSP (or BFOP) problem with $n$ variables
and $m$ constraints (or functions) with independent random weights in $[-1,1]$ with
density at most~$\phi$. Then with probability at least $1-o_n(1)$ over the draw of weights, any implementation of the FLIP 
algorithm running on $I$ 
takes at most $ \phi m  n^{O(\sqrt{\log n})}$ many steps to terminate.
\end{theorem}

\subsection{The rank-based approach}\label{sec:rank}

We briefly review the ranked-based approach of \cite{etscheid2015smoothed} and then discuss the main technical barrier
  to obtaining an upper bound that is asymptotically better than $n^{O(\log n)}$.

Since the maximum possible weight of a cut in the weighted graph is at most $O(n^2)$, if an execution of the FLIP algorithm is very long, then almost all the steps must have a very small gain, less than some small amount $\epsilon$.
Therefore, the execution must contain many long substrings  (consecutive subsequences) of moves, 
all of which yield very small gain, in $(0,\epsilon]$. 
Let   $B=(\sigma_1,\ldots,\sigma_{k})$ be a sequence of moves, where the $\sigma_i$'s are the nodes
flipped in each step, and
let $\gamma:V\rightarrow \{-1,1\}$ be the configuration (partition) of the nodes at the beginning.
The increase of the cut weight made by the $i$-th move  is a linear combination of the
weights of the edges incident to the node $\sigma_i$
that is flipped in the $i$-th step with coefficients either $-1$ or $1$; thus, the increase can be written
  as the inner product of a $\{-1,0,1\}$-vector indexed by $e\in E$ and the edge weight vector $X$.
We refer to the former as the improvement vector of the $i$-th move. From our assumption about the probability distributions of  edge weights, it is
easy to see that for any step, the probability that the increase lies in $(0,\epsilon]$ is at most $\phi \epsilon$. If these events for different steps were independent, then the probability that all the steps of the sequence have this property would be at most $(\phi \epsilon)^k$,
i.e., it would go down rapidly to 0 with the length $k$ of the sequence. Unfortunately these events may be quite correlated. However, a lemma of \cite{etscheid2015smoothed} (restated as Lemma \ref{lem:slowly_increase_probability} in Section \ref{sec:prelim})  shows that if the improvement vectors in some steps are linearly independent then they behave like independent events in the sense that the probability that they all yield a gain in $(0,\epsilon]$ is at most  $(\phi \epsilon)^r$, where $r$ is the number of linearly independent steps. This suggests that a useful parameter for obtaining a bound is the rank of the set of improvement vectors for the steps of the sequence.

One problem is that the improvement vectors 
generally depend on the initial configuration $\gamma$ of nodes that do not appear in the sequence $B$. Their number may be much larger than the rank $r$, and thus considering all their possible initial values will overwhelm the probability  $(\phi \epsilon)^r$. For this reason, \cite{etscheid2015smoothed} (and we) combine consecutive occurrences of the same node in the sequence $B$ of moves: for each pair $(i,j)$, $i<j\in [k]$, such that $\sigma_i$ and $\sigma_j$ are  two consecutive occurrences of the same node in $B$ (we call such a pair an {\em arc}), we form the improvement vector of the arc by summing the improvement vectors of the two steps $i$ and $j$. Thus, the total gain in cut weight from the two steps is given by the inner product of the improvement vector for the arc and~$X$; if every step of $B$ has gain at most $\epsilon/2$ then every arc has gain at most $\epsilon$. We call such a sequence $\eps$-\emph{improving}. The improvement vectors of the arcs do not depend on the initial configuration of   \emph{inactive} nodes, those that do not appear in the sequence. The \emph{rank} of the sequence $B$ is defined as the rank of the matrix $M_{B,\gamma}$ whose rows correspond to  edges of $G$ and whose columns are  improvement vectors of arcs of $B$.
The aforementioned lemma (Lemma \ref{lem:slowly_increase_probability} in Section \ref{sec:prelim}) then implies that if the rank of a sequence $B$ is $r$ then the probability that $B$ is $\eps$-improving is at most $(\phi \eps)^r$.

The main technical lemma of \cite{etscheid2015smoothed}, which we will refer to as \emph{the rank lemma}, shows that
\begin{center}\begin{minipage}{15cm}
\emph{Given any sequence $H$ of length $5n$,
  there always exists a substring $B$ of $H$ such that \\the rank\footnote{Note that the rank 
  is defined earlier using both $B$ and the initial configuration $\gamma$. 
An observation from \cite{angel2017local} shows that the rank actually does not depend on $\gamma$ but only $B$.}
of $B$ is at least $\Omega(\leng(B)/\log n)$.}
\end{minipage}\end{center}
\clearpage 
With this lemma, one can apply a union bound to upper bound the probability that there are an initial configuration $\gamma$ and a sequence $B$ with $\leng(B)\le 5n$ and rank $\Omega(\leng(B)/\log n)$ such that $B$
is $\eps$-improving with respect to $\gamma$ and $X\sim\calX$ as follows:~\begin{equation}\label{hauu}
\sum_{\ell\in [5n]} 2^\ell \cdot n^{\ell}\cdot (\phi \eps)^{\Omega(\ell/\log n)}.
\end{equation}
Here $n^\ell$ is a trivial bound for the number of sequences of length $\ell$
  and $(\phi\eps)^{\Omega(\ell/\log n)}$ is the probability that a sequence with rank 
  $\Omega(\ell/\log n)$ is $\eps$-improving.
A crucial observation is that, because of the definition of ranks (based on arcs instead of individual moves),
  we do not need to apply the union bound on the $2^n$ configurations over all nodes but
  only on configurations of nodes that appear in the sequence.
In other words, initial configurations that only differ on non-active nodes can be treated as the same.
This is why we can use $2^\ell$ instead of $2^n$ in (\ref{hauu}) since $\ell$ is a trivial upper bound for the number of active nodes.
By setting $\eps=1/(\phi n^{O(\log n)})$, (\ref{hauu}) becomes $1-o_n(1)$.
It follows from the rank lemma that, with high probability,
  no sequence $H$ of length $5n$ can be $\eps$-improving and thus, 
  the cut weight must go up by at least $\eps$ for every $5n$ moves.
The $\phi n^{O(\log n)}$ upper bound of \cite{etscheid2015smoothed} then follows since the
maximum possible weight of a cut is $O(n^2)$.

A natural question for further improvements is whether  the $\log n$-factor
  lost in the rank lemma of \cite{etscheid2015smoothed} is necessary. 
Taking a closer look, the proof of \cite{etscheid2015smoothed} consists of two steps. First it is shown that given any
  sequence $H$ of length $5n$, there is a substring $B$ such that the number of repeating
  nodes in $B$ (i.e., those that appear at least twice in $B$) is 
  $\Omega(\leng(B)/\log n)$.
The rank lemma then follows by showing that the rank of $B$ is at least proportional to the number of repeating nodes in it
  (which we include as Lemma \ref{firstlowerbound} in Section \ref{sec:basiclemma}).
On the one hand, the first step of the proof turns out to be tight given an
  example constructed in \cite{angel2017local}.
Furthermore, we give a construction in Appendix \ref{sec:hard-case} to show that, not only the proof approach of 
  \cite{etscheid2015smoothed} is tight, but the rank lemma itself is indeed tight,
  by giving a graph $G$ and a sequence $H$ of length $5n$ such that 
  every substring $B$ of $H$ has rank at most $O(\leng(B)/\log n)$. 
Therefore, one cannot hope to obtain a bound better than $n^{O(\log n)}$ based on an
  improved version of this rank lemma.

\subsection{A new rank lemma}\label{sec:newrank}

We overcome the $\log n$-barrier to the rank-based approach of \cite{etscheid2015smoothed} on general graphs by considering 
  not only substrings of $H$ but also its \emph{subsequences}.
Recall that a subsequence of $H$ is of the form $(\sigma_{i_1},\ldots,\sigma_{i_k})$ with
  $i_1<\cdots<i_k$.
We use the same arc-based rank notion defined above. The main technical component (Lemma \ref{lem:main}) is a new rank lemma that 
  can be stated informally as follows:
\begin{center}\begin{quote}
\emph{If $H$ is a sequence of moves of length $5n$, then there is a subsequence $B$ of $H$\\
  such that the rank of $B$ is at least $\Omega(\leng(B)/\sqrt{\log n})$.}
\end{quote}\end{center}
  
While the $\sqrt{\log n}$ in the statement naturally leads to the improvement from
  $\log n$ to $\sqrt{\log n}$ in our smoothed complexity bound,
  one needs to be careful when working with subsequences $B$ of $H$.  
An advantage of using substrings of $H$ is that improvement vectors of arcs
  are trivially preserved, which is not necessarily the case for subsequences of $H$.
More formally, let $B=(\sigma_\ell,\ldots,\sigma_r)$ be a substring of $H$
  and $\alpha=(i,j)$ be an arc of $H$ such that $\ell\le i<j\le r$.
Then the corresponding arc $\beta=(i-\ell+1,r-\ell+1)$ of $B$ has the same 
  improvement vector as that of $\alpha$ in $H$.
Therefore, $B$ being not $\eps$-improving trivially implies that $H$ is not $\eps$-improving.
However, when $B=(\sigma_{i_1},\ldots,\sigma_{i_k})$ is a subsequence of $H$,
  every arc $\beta$ of $B$ can be mapped back to be an arc $\alpha$ of $H$ but
  now it is in general not true that $\alpha$ and $\beta$ share the same improvement vector and thus,
  $B$ being not $\eps$-improving does not
  necessarily imply that $H$ is not $\eps$-improving. 
  
{{\color{blue} } Despite this limitation, we  prove a \emph{subsequence rank lemma} in Section \ref{sec:sqrtn-ub} of the following form }(still an informal\footnote{The lemma
  stated here is still not in its  formal version since we ignore the 
  involvement of the initial configuration $\gamma$; see Lemma \ref{lem:main} for details.
Fortunately the initial configuration will play a minimal role in the proof and 
  we find it easier to gain intuition about the proof without considering it in the picture.}
   version; see Lemma \ref{lem:main}):
\begin{center}\begin{minipage}{13.5cm}
\emph{If $H$ is a sequence of moves of length $5n$, then there is a subsequence $B$ of $H$
  and a set of arcs $Q$ of $B$ such that the rank of $Q$ (i.e., the rank of the matrix where
  we only include improvement vectors of arcs in $Q$) is at least $\Omega(\leng(B)/\sqrt{\log n})$ and
  the improvement vector of every arc in $Q$ is the same as that of its corresponding arc in $H$.}
\end{minipage}\end{center}
$\text{\\}$Theorem \ref{maintheorem} then follows quickly from the new rank lemma by a similar union bound. 
$\text{\\}$

The technical challenge for proving our new rank lemma is to balance the following trade-off.
On the one hand, 
  we would like to keep as many arcs of $H$ in $Q$ as possible so that they together give us 
  a high rank compared to the length of $B$.
On the other hand, the more arcs we want to keep the less aggressively we can
  delete moves from $H$, in order to have their improvement vectors preserved.
To achieve this for an arc $\alpha=(i,j)$ of $H$, we need to make sure that the parity of the number of occurrences inside the arc of any node 
  adjacent to the node $\sigma_i=\sigma_j$ in $G$  remains the same after deletions.

We now give a sketch of the proof of our Main Lemma (Lemma \ref{lem:main}). 
Let $H$ be a sequence of moves of length $5n$. Given that it is much longer than the number $n$ of vertices,
  it is easy to show that $H$ has many arcs (actually at least $4n$; see Lemma \ref{lem:arc-length-ratio-lb}).
We first partition all arcs of $H$ into $\log n$ many chunks according to their lengths (the length
  of an arc $(i,j)$ is defined to be $j-i+1$): chunk $\chunk_j$ contains all arcs of length between $2^j$ and 
  $2^{j+1}$.
Then there must be a $j^*$ such that $|\chunk_j^*|$ is at least $\Omega(n/\log n)$.
Focusing on arcs in $\chunk_{j^*}$ and letting $\smash{\ell=2^{j^*+1}}$, one can show (Lemma \ref{lem:arc_node_ratio_in_cover} in Section \ref{sec:basiclemma})
  that there is a substring $H'=(\sigma_i,\ldots,\sigma_{i+2\ell-1})$ of length $2\ell$
  such that the number of $\chunk_{j^*}$-arcs contained in $H'$ is $\Omega(\ell/\log n)$
  (this should not come as a surprise because this is basically the expected number of $\chunk_{j^*}$-arcs
  when we pick the window uniformly at random).
Let $C$ be the set of $\chunk_{j^*}$-arcs in $H'$.
If we take $B$ to be $H'$ and $Q$ to be arcs that correspond to $C$ in $B$, then
  the rank of $Q$ can be shown to be  $\Omega(|Q|)$ (by applying Lemma \ref{firstlowerbound} 
  discussed earlier
  and using the fact that all arcs in $Q$ are almost as long as $B$ up to a constant).
However, the ratio $|Q|/\leng(B)=|C|/(2\ell)$ is only $\Omega(1/\log n)$, too weak for our goal.
Instead our proof uses the following new ideas.

The first idea is to group the $\log n$ chunks
  $\chunk_1,\ldots,\chunk_{\log n}$ into $\sqrt{\log n}$ groups $\smash{\group_1,\ldots,\group_{\sqrt{\log n}}}$, each
  being the union of $\sqrt{\log n}$ consecutive chunks.
In Case 1 and Case 2  of the proof, we pick a group $\group_{i^*}$, with $\ell''$ set to be
  the maximum length of arcs in $\group_{i^*}$, and then pick a substring $H''$ of $H$ of length $2\ell''$ by
  Lemma \ref{lem:arc_node_ratio_in_cover} so that the number
  of $\group_{i^*}$-arcs in $H''$ is  
  $\Omega(\ell''/\sqrt{\log n})$.
We show that when these $\group_{i^*}$-arcs satisfy
  certain additional properties (see more discussion about these properties below), then their rank is
  almost full and Lemma \ref{lem:main} for these two cases follows by setting $B$ to be $H''$ and $Q$ to be arcs of $B$ that correspond to these $\group_{i^*}$-arcs 
  in $H''$.
  
 The second idea is to continue using the 
  substring $H'$ and the set $C$ of $\chunk_{j^*}$-arcs
  in it, with the rank of $C$ being $\Omega(\leng(H')/\log n)$, but now we try to \emph{delete} as many moves from $H'$ as possible to obtain the desired subsequence $B$ and at the same time \emph{preserve} improvement vectors of arcs in $C$.
  
We make two key observations about which moves can or cannot be deleted.
First let $\sigma_k$ be a move in $H'$ such that node $\sigma_k$ only appears once in $H'$.
Then we cannot delete $\sigma_k$ if $i<k<j$ for some arc $\alpha=(i,j)\in C$ 
  and $(\sigma_i,\sigma_k)$ is an edge in $G$; otherwise the improvement 
  vector of $\alpha$ will not remain the same at the entry indexed by edge $(\sigma_i,\sigma_k)$.
As a result, if there are many such moves in $H'$ then we cannot hope to preserve arcs in $C$
  and at the same time increase the ratio $|C|/\leng(B)$ up to $1/\sqrt{\log n}$.
To handle this situation, our first key observation is that having many such $\sigma_k$ is indeed
  a good case: it would imply that many arcs $\alpha=(i,j)$ in $H$ have a 
  $\sigma_k$ (referred to as a \emph{witness} for $\alpha)$ such that $i<k<j$, $(\sigma_i,\sigma_k)$ is an edge in $G$,
  $\sigma_k$ only appears once inside $\alpha$, and \emph{both the previous and next occurrences
  of $\sigma_k$ are pretty far away from $k$}. We handle this case in Case 2 of our proof.
As discussed earlier, we pick a group $\group_{i^*}$ and 
  a substring $H''$ of $H$.
Assuming that most $\group_{i^*}$-arcs in $H''$
  satisfy this additional property now,
  their witnesses can then be used to certify
  the linear independence of their improvement
  vectors; this  implies
  that these $\group_{i^*}$-arcs in $H''$ have 
  almost full rank.

The next observation is about repeating nodes in $H'$.
Let $\beta=(k,r)$ be an arc that shares no endpoint with arcs in $C$.
We say $\beta$ \emph{overlaps} with an arc $\alpha=(i,j)\in C$ if the $(\sigma_k,\sigma_i)$ is
  an edge in $G$ and either $k<i<r<j$ or $i<k<j<r$.
If $\beta$ does \emph{not} overlap with any arc in $C$ then it is not difficult to show
  that the deletion of both moves $k$ and $r$ of $\beta$ 
  will have no effect on improvement vectors
  of arcs in $C$.
Therefore, we can keep deleting until no such arc exists in $H'$ anymore.
But, what if many arcs in $H'$ overlap with arcs in $C$?  
Our second observation is that this is again a good case for us. 
Assuming that there are $\Omega(\ell/\sqrt{\log n})$ arcs in $H'$ that overlap with arcs in $C$,
  we show that the rank of these arcs is almost full and 
  thus, the ratio of the rank and the length of $H'$ is $\Omega(1/\sqrt{\log n})$;
  this is our Case 3.1.
(Note that the discussion here is very informal. In the actual proof, we need to 
  impose an extra condition (see Definition \ref{def:good}) on arcs in $C$ in order to show that 
  the rank of arcs overlapping with arcs in $C$ is almost full. 
We handle the case when most arcs of $H$ violate this condition in Case 1 of the proof, by working with a group
  $\group_{i^*}$ as discussed earlier.)
  
Now we can assume that all moves in $H'$ can be deleted except    those that are endpoints of arcs in $C$ and endpoints of arcs that overlap with at least
  one arc in $C$ (the number of which is at most $O(\ell/\sqrt{\log n})$).
Recall from the discussion at the beginning that the rank of $C$ is almost full.
Given that the length of the subsequence $B$ obtained after deletions 
  is  $O(\sqrt{\log n})\cdot |C|$, the rank lemma follows
  (since we made sure that the deletion of moves does not affect improvement vectors of arcs in $C$).
This is handled as the last case, Case 3.2, in the proof of the Main Lemma.

With the proof sketch given above, the choice of $\sqrt{\log n}$ in the statement of the Main Lemma is 
   clearly the result of balancing these delicate cases.
At a high level, the proof of the Main Lemma 
relies on a detailed classification of arcs
  based on a number of their attributes that we can take advantage in
  the analysis of their ranks.
The proof involves an intricate analysis of sequences and their properties and uses very little from the structure of the graph itself and the Max-Cut problem. As a consequence, the proof readily extends to all other local Max-2CSP problems with the same bound on their smooth complexity.\medskip

\noindent\textbf{Organization.} The structure of the rest of the paper is as follows. Section \ref{sec:prelim} gives basic definitions and background. Section \ref{sec:maintheorem} states the Main Lemma and uses it to prove Theorem \ref{maintheorem}. Section \ref{sec:sqrtn-ub}, which is technically the heart of the paper, proves the Main Lemma. Section \ref{sec:maxcsp} presents the extension to general binary Max-CSP and Function problems, and
Section \ref{sec:conclusions} offers concluding remarks and open problems.

\section{Preliminaries}\label{sec:prelim}

Given a positive integer $n$ we use $[n]$ to denote $\{1,\ldots,n\}$.
Given two integers $i\le j$, we write $[i:j]$ to denote the interval
  of integers $\{i, \ldots,j\}$. {Given an interval $I=[i:j]$, we write $\leng(I)=j-i+1$ to denote the \emph{length} of the interval $I$.}

Let $G=(V,E)$ be a weighted undirected graph with a weight vector
  $X=(X_e:e\in E)$, where $X_e\in [-1,1]$ is the weight of edge $e\in E$.
Under the smoothed complexity model,
  there is a family  $\calX=(\calX_e: e\in E)$ of probability distributions, one for each edge;
  the edge weights $X_e$ are drawn independently from the corresponding distributions $\calX_e$.
We assume that each $\calX_e$ is a distribution supported on $[-1,1]$ and its density function is bounded from above by a parameter $\phi>0$. (The assumption that the edge weights are in $[-1,1]$ is no loss of generality, since they can be always scaled to lie in that range.) 
A \emph{configuration} $\gamma$ of a set of nodes $S\subseteq V$ is a map from $S$ to $\{-1,1\}$. A configuration $\gamma$ of $V$ corresponds to a partition of the nodes into two parts: the left part $\{ u \in V: \gamma(u)=-1\}$ and the right part $\{ u \in V: \gamma(u)=1\}$. The \emph{weight} of a configuration (partition) $\gamma$ of $V$ with respect to a weight vector $X$ is the weight of the corresponding cut, i.e., the sum of weights of all edges that connect a left node with a right node.

 Formally, it is  given by
\begin{equation}\label{eqn:maxcut-obj}
\textsf{obj}_{G,X}(\gamma)= \sum_{(u,v)\in E}  {X_{(u,v)}}\cdot \mathbf{1}\{\gamma(u) \neq \gamma(v) \} = \frac{1}{2}\sum_{(u,v)\in E} X_{(u,v)}\cdot \big(1-\gamma(u)\gamma(v)\big).
\end{equation}
The problem of finding a configuration of $V$ that maximizes the cut weight is the well-known~Max-Cut problem.
We are interested in the Local Max-Cut problem, where
  the goal is to find a configuration $\gamma$ of $V$ that is 
  a local optimum, i.e., 
  $\textsf{obj}_{G,X}(\gamma)\ge \textsf{obj}_{G,X}(\gamma^{(v)})$
  for all $v\in V$, where $\gamma^{(v)}$ is the configuration
  obtained from $\gamma$ by flipping the sign of $\gamma(v)$.
  
A simple algorithm for Local Max-Cut  is the following FLIP algorithm.\medskip
\begin{flushleft}\begin{minipage}{16cm}
\emph{``Start from some initial configuration $\gamma=\gamma_0$ of $V$.  While there exists a node $v\in V$
such that flipping the sign of $\gamma(v)$ would increase the cut weight, select such a node $v$
(according to some pivoting criterion) and execute the flip, i.e., set $\smash{\gamma_{i+1}=\gamma_i^{(v)}}$ and repeat.}\medskip
\end{minipage}\end{flushleft}
The algorithm terminates with 
  a configuration of $V$ that cannot be improved by flipping any single node. The execution of FLIP for a given graph $G$ and edge weights $X$ depends on both the initial configuration $\gamma_0$ and the pivoting criterion used to select a node to flip in each iteration, when there are multiple nodes which can be profitably moved. Each execution of FLIP generates a sequence of nodes that are moved during the execution.

Given $G=(V,E)$  {we denote a \emph{sequence of moves} as a sequence $H=(\sigma_1,\ldots,\sigma_k)$
  of nodes from $V$, where we write $\leng(H)=k$ to denote its \emph{length}.}
We say a node $v\in V$ is \emph{active} in $H$ if it~appears in $H$, and is 
  \emph{repeating} if it appears at least twice in $H$.
We write $S(H)$ to denote the set of active nodes in $H$,
and use $S_1(H)$ (resp. $S_2(H)$) to denote the set of nodes that appear only once (resp. two or more times) 
 in $H$.
As usual, a \emph{substring} of $H$ is a sequence of the form~$(\sigma_i, \sigma_{i+1},\ldots,\sigma_j)$ for some
  $1\le i<j\le k$,  and a \emph{subsequence} of $H$ is 
  a sequence of the form $(\sigma_{i_1},\ldots,
  \sigma_{i_\ell})$ for some $1\le i_1<\cdots<i_\ell\le k$.
Given a set $P\subseteq [k]$, we write $H_P$
  to denote the subsequence of $H$ obtained by restricting to indices in $P$.
When $P$ is an interval $[i:j]\subseteq [k]$, $H_P$ is a substring of $H$.

Next we introduce the notion of arcs and define their improvement vectors. 
An \emph{arc} $\alpha=(i,j)$ of $H=(\sigma_1,\ldots,\sigma_k)$ is a pair of indices
$i<j\in [k]$ such that $\sigma_i=\sigma_j$ and $\sigma_i\ne\sigma_\ell$
for all $i<\ell<j$ 
(i.e., $\sigma_i$ and $\sigma_j$ are two 
consecutive occurrences of the same node in $H$).
We let $\vertex_H(\alpha)=\sigma_i=\sigma_j\in V$ and refer to it as the node of $\alpha$.
We also refer to $i$ as the \emph{left endpoint} and $j$ as the \emph{right endpoint} of $\alpha$,
  and write $\leftt(\alpha)=i$ and $\rightt(\alpha)=j$.
  We will sometimes omit the subscript $H$ when it is clear from the context.
We write $\leng(\alpha)=j-i+1$ to denote the length of $\alpha$.

Given a sequence $H=(\sigma_1,\ldots,\sigma_k)$ of moves (nodes) and an initial configuration $\gamma=\gamma_0$ before the first move of $H$, let $\gamma_i$ denote the configuration after the $i$-th move of $H$. The gain in the cut weight from the $i$-th move is a linear combination of the weights of the edges incident to node $\sigma_i$ that is flipped, where some edges have coefficient $1$ and the rest have coefficient $-1$. Note that if $H$ is part of an execution of the FLIP algorithm, then the gain is positive at every move.

For each arc $\alpha=(i,j)$ of $H$, we define
  the \emph{improvement vector}~of $\alpha$ with respect to $\gamma$ and $H$, denoted by 
  $\improvement{\gamma,H}{\alpha}$, as follows:
$\improvement{\gamma,H}{\alpha}$ is a vector in $\{-2,0,2\}^{ E }$ indexed by edges $e\in E$ (just like  the weight vector $X$); 
its entry indexed by $e\in E$ is
nonzero iff $e=(\vertex_H(\alpha),v)\in E$~for some node $v$ that appears
  an odd number of times in $\sigma_{i+1},\ldots,\sigma_{j-1}$. 
When this is the case, 
  its value is set to be $2\gamma_{i-1}(\vertex_H(\alpha))\gamma_{i-1}(v)$. 
{Note that for this definition we do not need to have the full configuration $\gamma$
  of $V$ but only of the active nodes in $S(H)$.
It also follows from the definition of improvement vectors that, if $\gamma$ is the 
  initial configuration of $S(H)$  and we move nodes one by one 
  according to $H$,  then the \emph{total} gain in the cut weight $\textsf{obj}$
  from the $i$-th move and the $j$-th move is given by the inner product of~$\improvement{\gamma,H}{\alpha}$~and~$X$.
  Indeed, letting $u=\vertex_H(\alpha)$, the total gain from these two moves equals
  \[\underbrace{\sum_{v:\hspace{0.03cm} (u,v) \in E} X_{(u,v)}\cdot \gamma_{i-1}(u)\gamma_{i-1}(v) }_{i\text{-th move}}+
 \underbrace{\sum_{v:\hspace{0.03cm} (u,v)\in E} X_{(u,v)}\cdot \gamma_{j-1}(u)\gamma_{j-1}(v) }_{j\text{-th move}}
 \]
Given that $\gamma_{i-1}(u)=-\gamma_{j-1}(u)$, only those neighbors $v$ of $u$ that flipped an odd number of times~in
  $\sigma_{i+1},\ldots,\sigma_{j-1}$, i.e., $\gamma_{i-1}(v)\neq \gamma_{j-1}(v) $, contribute in the \emph{total} gain of the two moves.}
  
  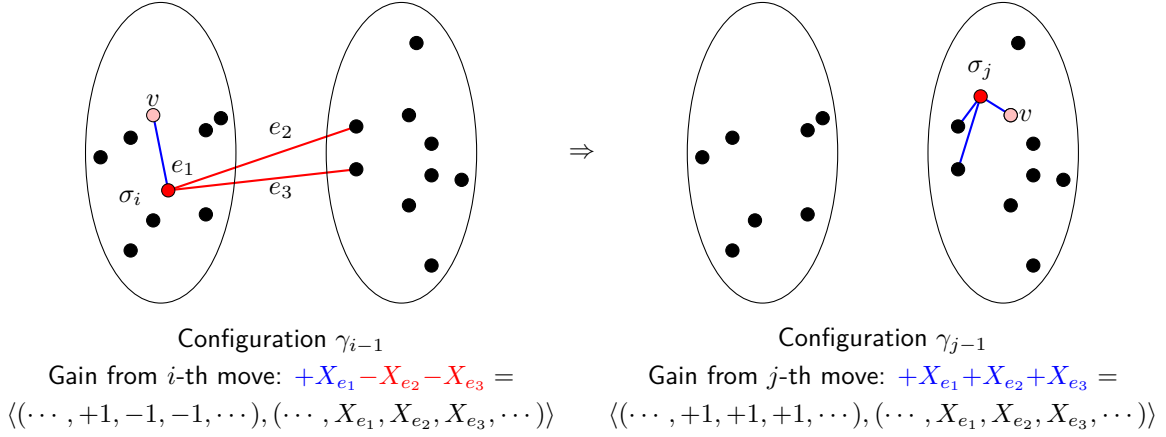
\begin{figure}[t!]
\begin{tikzpicture}
\tikzstyle{every node}=[node font={\sf \small}]

\draw[thick,draw=blue] (-1.5,-0.5) -- (-1.7,0.5);
\draw[thick,draw=red] (-1.5,-0.5) --  (1,0.35);
\draw[thick,draw=red] (-1.5,-0.5) --  (1,-0.22);
\draw[thick,draw=blue] (8+1.3,0.75) -- (8+1.7,0.5);
\draw[thick,draw=blue] (8+1.3,0.75) --  (8+1,0.35);
\draw[thick,draw=blue] (8+1.3,0.75) --  (8+1,-0.22);

\draw[fill=black] (-1.7,-0.9) circle (2.5pt);
\draw[fill=black] (-2.4,-0.06) circle (2.5pt);
\draw[fill=black] (-1,-0.82) circle (2.5pt);
\draw[fill=black] (-2,-1.3) circle (2.5pt);
\draw[fill=red] (-1.5,-0.5) circle (2.5pt);
\draw[fill=pink] (-1.7,0.5) circle (2.5pt);
\draw[fill=black] (-0.8,0.46) circle (2.5pt);
\draw[fill=black] (-2,0.2) circle (2.5pt);
\draw[fill=black] (-1,0.3) circle (2.5pt);

\draw[fill=black] (1.7,-0.7) circle (2.5pt);
\draw[fill=black] (2.4,-0.36) circle (2.5pt);
\draw[fill=black] (1,-0.22) circle (2.5pt);
\draw[fill=black] ( 2,-0.3) circle (2.5pt);
\draw[fill=black] (2,-1.5) circle (2.5pt);
\draw[fill=black] (1.7,0.5) circle (2.5pt);
\draw[fill=black] (1.8,1.46) circle (2.5pt);
\draw[fill=black] (2,0.12) circle (2.5pt);
\draw[fill=black] (1,0.35) circle (2.5pt);

\draw (-1.6,0) ellipse (1cm and 2cm);
\draw (1.6,0) ellipse (1cm and 2cm);

\node at (0,-2.5) {Configuration $\gamma_{i-1}$};
\node at (-1.3,-0.2) {$e_1$};
\node at (-2.0,-0.6) {$\sigma_i$};


\draw[fill=black] (8-1.7,-0.9) circle (2.5pt);
\draw[fill=black] (8-2.4,-0.06) circle (2.5pt);
\draw[fill=black] (8-1,-0.82) circle (2.5pt);
\draw[fill=black] (8-2,-1.3) circle (2.5pt);
\draw[fill=red]   (8+1.3,0.75) circle (2.5pt);
\draw[fill=black] (8-0.8,0.46) circle (2.5pt);
\draw[fill=black] (8-2,0.2) circle (2.5pt);
\draw[fill=black] (8-1,0.3) circle (2.5pt);

\draw[fill=black] (8+1.7,-0.7) circle (2.5pt);
\draw[fill=black] (8+2.4,-0.36) circle (2.5pt);
\draw[fill=black] (8+1,-0.22) circle (2.5pt);
\draw[fill=black] (8+2,-.3) circle (2.5pt);
\draw[fill=black] (8+2,-1.5) circle (2.5pt);
\draw[fill=black] (8+1.7,.5) circle (2.5pt);
\draw[fill=black] (8+1.8,1.46) circle (2.5pt);
\draw[fill=black] (8+2,0.12) circle (2.5pt);
\draw[fill=black] (8+1,0.35) circle (2.5pt);
\draw[fill=pink]  (8+1.7,0.5) circle (2.5pt);

\draw (8-1.6,0) ellipse (1cm and 2cm);
\draw (8+1.6,0) ellipse (1cm and 2cm);
\node at (4+0,0) {$\Rightarrow $};
\node at (8+0,-2.5) {Configuration $\gamma_{j-1}$};
\node at (8+1.3,1.1) {$\sigma_j$};

\node at (0,0.3) {$e_2$};
\node at (0,-0.5) {$e_3$};
\node at (+0,-3) {Gain from $i$-th move: ${\color{blue}+X_{e_1}}{\color{red}-X_{e_2}}{\color{red}-X_{e_3}}=$};
\node at (+0,-3.5) { $\left \langle (\cdots,+1,-1,-1,\cdots),(\cdots,X_{e_1},X_{e_2},X_{e_3},\cdots) \right \rangle$};
\node at (8+0,-3) {Gain from $j$-th move: ${\color{blue}+X_{e_1}}{\color{blue}+X_{e_2}}{\color{blue}+X_{e_3}}=$};
\node at (8+0,-3.5) {$\left \langle (\cdots,+1,+1,+1,\cdots),(\cdots,X_{e_1},X_{e_2},X_{e_3},\cdots) \right \rangle$};

\node at (-1.7,0.7) {$v$};
\node at (8+1.9,0.5) {$v$};

\end{tikzpicture}
\centering
\caption{An illustration of an arc $\alpha=(i,j)$. Adding the two corresponding improvement vectors\newline results in $(+2,0,0)$, since only the neighbor $v$ moved an odd number of times in the meantime.}
\end{figure}

Inspired by the definition of improvement vectors above, we define~the \emph{interior} of $\alpha$, denoted by $\interior(\alpha)$,
  as follows: $\interior(\alpha)$ contains all $k\in [i+1:j-1]$ such that 
  node $\sigma_k$ appears an odd number of times in $\sigma_{i+1},\ldots,\sigma_{j-1}$
  and $\sigma_k$ is adjacent to $\sigma_i=\vertex_H(\alpha)$ in the graph $G$. 
  
We say an arc $\alpha$ of $H$ is 
  \emph{improving} with respect to $\gamma$ and $X$ if
    the inner product of $\improvement{\gamma,H}{\alpha}$ and $X$
   is positive.
We say it is \emph{$\eps$-improving} for some parameter $\eps>0$ if the inner product is~in~$(0,\eps]$.
Furthermore, we say a set $C$ of arcs of $H$ 
  is improving (or $\eps$-improving)  with respect to $\gamma$ and $X$ if
  every arc in $C$ is  improving (or $\eps$-improving).
A sequence $H$ of moves is improving (or $\eps$-improving)
  with respect to $\gamma$ and $X$ if every arc of $H$ is improving (or $\eps$-improving).
Note that if $H$ is part of the sequence of moves generated by an execution of the FLIP algorithm then $H$~must~be improving, because every move must increase the weight of the cut and therefore every arc is improving. 
On the other hand,  
 if some move in $H$ increases the cut weight by more than $\epsilon$ then the same is true for the arc that has it as an endpoint and thus, 
 $H$ is not $\epsilon$-improving.

Let $C$ be a set of arcs of $H$. 
A key parameter of $C$ that will be used~to bound the probability of $C$ being $\eps$-improving
  (over the draw of the edge weights $X\sim \calX$) is the \emph{rank} of $C$ with respect to $\gamma$ and $H$, denoted by
  $\rank_{\gamma,H}(C)$: this is the rank of the $|E|\times |C|$ matrix that contains 
  improvement vectors $\improvement{\gamma,H}{\alpha}$ as its column vectors, one for
  each arc  $\alpha\in C$.
To give some intuition for this parameter, one may hope that for 
a fixed sequence of moves $H$ with $K$ arcs $$\Pr_{X\sim \calX}\Big[\hspace{0.04cm}\text{$H$ is $\eps$-improving}\hspace{0.04cm}\Big]=
\prod_{i=1}^{K}
\Pr_{X\sim \calX}\Big[\hspace{0.04cm}\text{$i$-th arc of $H$ is $\eps$-improving}\hspace{0.04cm}\Big].$$
However, since there could be improving steps that are strongly correlated
  (as an extreme situation there could be two arcs with exactly the same improvement vector), one may expect the product on the right hand side to hold 
 only for linearly independent  
 $\improvement{\gamma,H}{\alpha}$'s, introducing the necessity of analysis of the $\rank_{\gamma,H}(C)$.

An observation from \cite{etscheid2015smoothed}  is that $\rank_{\gamma,H}(C)$ is \emph{independent} of the choice of $\gamma$. 
{{\color{blue}} Indeed, a change of a node in the initial configuration would result in a change of sign on every row of the matrix that is incident with this node.}
So from now on we write it as $\rank_H(C)$.
To simplify our discussion on $\rank_H(C)$ later, we use
  $\improvement{H}{\alpha}$ to denote $\improvement{\gamma_0,H}{\alpha}$
  where $\gamma_0$ is the default initial configuration of $S(H)$ that maps every node in $S(H)$ to $-1$.
Then $\rank_H(C)$ is the rank of the matrix that consists of   $\improvement{H}{\alpha}$,
  $\alpha\in C$.
The next tool from \cite{etscheid2015smoothed} shows that the higher the rank is, the 
  less likely for $C$ to be $\eps$-improving.
  
\begin{lemma}[Lemma A.1 from \cite{etscheid2015smoothed}]\label{lem:slowly_increase_probability}
Let $\calX=(\calX_i:i\in [m])$ be a sequence of probability distributions in which each
  $\calX_i$ has density bounded from above by a parameter $\phi>0$.
Let $r_1,\ldots,r_k\in \mathbb{Z}^m$ be $k$ vectors that are linearly independent.
Then for any $\eps>0$, we have
\[
\Pr_{X\sim \calX}\Big[\hspace{0.05cm}\forall\hspace{0.05cm} i\in [k]: \langle r_i,X\rangle \in [0,\eps] \hspace{0.04cm}\Big] \le (\phi\epsilon)^k.
\]
\end{lemma}

\begin{corollary}\label{coro1}
Let $G=(V,E)$ be an undirected graph and let $\calX=(\calX_e:e\in E)$ be a sequence
  of distributions such that each $\calX_e$ has density bounded from above 
  by a parameter $\phi>0$.
Let $H$ be a sequence of moves, $\gamma$ be a configuration of $S(H)$, and $C$ be a set of arcs of $H$.
Then for any $\eps>0$, 
$$
\Pr_{X\sim \calX}\Big[\hspace{0.04cm}\text{$C$ is $\eps$-improving   with respect to $\gamma$ and $X$}\hspace{0.04cm}\Big]
\le (\phi \eps)^{\rank_H(C)}.
$$
\end{corollary}

Finally we say a sequence $H$ of moves is \emph{nontrivial} if
  the interior of every arc in $H$ is nonempty;
  $H$ is \emph{trivial} if at least one of its arcs has $\interior(\alpha)=\emptyset$.
It follows from definitions that $H$ cannot be improving 
  (with respect to any $\gamma$ and any $X$) if it is trivial.
  Since every sequence resulting from an execution of the FLIP algorithm is improving, it follows that it is also nontrivial.
We will only consider henceforth only nontrivial sequences (see Lemma \ref{lem:main}).

\section{Main Lemma and the Proof of Theorem \ref{maintheorem}}\label{sec:maintheorem}

We prove the following main technical lemma in Section \ref{sec:sqrtn-ub}:

\begin{lemma}\label{lem:main}
Let $G=(V,E)$ be an undirected graph over $n$ vertices.
Given any  nontrivial sequence $H$ of moves of length $5n$ and any configuration $\gamma$ of $S(H)$,
  there exist (i) a sequence $B$ of moves of length at most $5n$,
  (ii) a configuration $\tau$ of $S(B)$, and (iii) a set of arcs $Q$ of $B$ such that
\begin{enumerate}
    \item The rank of $Q$ in $B$ satisfies 
    \begin{equation}\label{aslarge}\frac{\rank_B(Q)}{\leng(B)}\ge \Omega\left(\frac{1}{\sqrt{\log n}}\right) \quad\ \  (\text{High-rank property});
    \end{equation}\item For every arc $\alpha\in Q$, there exists an arc $\alpha'$ of $H$ such that $$\improvement{\tau,B}{\alpha}=\improvement{\gamma,H}{\alpha'}.
    \quad\ \  (\text{Vector-Preservation property}).
    $$
\end{enumerate}
\end{lemma}

As discussed earlier in Section \ref{sec:newrank}, the new sequence
  $B$ in Lemma \ref{lem:main} is either a substring  or a subsequence of $H$.
When we pick $B$ to be a substring of $H$, say a substring that starts with
  the $i$-th move of $H$, the natural choice of $Q$ is the set of all arcs in $B$ (since we would like $\rank_B(Q)$ to be as large as possible in (\ref{aslarge}))
  and that of $\tau$ is  the configuration $\gamma_{i-1}$ of $S(B)$
  derived from $\gamma$ after making the first $i-1$ moves of $H$.
With these choices, the second condition of Lemma \ref{lem:main} is trivially satisfied and
  the main goal is to lowerbound the rank of  arcs in $B$.
This is indeed the proof strategy followed in all previous works \cite{etscheid2015smoothed,angel2017local,bibak2019improving}.
The key new idea of the paper is the use of subsequences of $H$ as $B$ instead of substrings of~$H$. 
While this gives us more flexibility
  in the choice of $B$ to overcome the $(\log n)$-barrier of \cite{etscheid2015smoothed} as sketched earlier in Section \ref{sec:newrank},
  one needs to be very careful when deleting moves and picking arcs
  to be included in $Q$ in order to satisfy the second condition.
  

We delay the proof of Lemma \ref{lem:main} to Section \ref{sec:sqrtn-ub}.
Instead, below we  use it to prove Theorem \ref{maintheorem}.

\maintheorem*
\begin{proof}[Proof of Theorem \ref{maintheorem} assuming Lemma \ref{lem:main}]
Let $c_1>0$ be a constant to be specified later and let 
$$\eps=\frac{1}{\phi\cdot n^{c_1\sqrt{\log n}}}.$$
We write $F$ to denote the following event on the draw of the
  weight vector $X\sim \calX$:
\begin{flushleft}\begin{quote}
Event $F$: For every sequence $B$  of length at most $5n$, every configuration
  $\tau$ of $S(B)$, and every set $Q$ of arcs of $B$ satisfying
  (where $a>0$ is the constant in Lemma \ref{lem:main})
\begin{equation}\label{hehe1}
\frac{\rank_B(Q)}{\leng(B)}\ge  \frac{a}{\sqrt{\log n}},
\end{equation}
$Q$ is \emph{not} $\eps$-improving with respect to $\tau$ and $X$.
\end{quote}\end{flushleft}
We break the proof into two steps.
First we prove that $F$ occurs 
  with probability at least $1-o_n(1)$ over the draw of the weight vector
   $X\sim \calX$.
Next we show that when $F$ occurs, any implementation of 
  the FLIP algorithm must terminate within $\phi\cdot n^{O(\sqrt{\log n})}$ many
  steps.

For the first step, we fix an $\ell\in [5n]$, a sequence $B$ of length $\ell$,
  a configuration $\tau$ of $S(B)$ and a set $Q$ of arcs of $B$ that satisfies (\ref{hehe1})
  (so the rank is at least $a\ell/\sqrt{\log n}$).
It follows from Corollary \ref{coro1} that the probability of $Q$ being $\eps$-improving
  with respect to $\tau$ and $X\sim \calX$ is at most $\smash{(\phi \eps)^{a\ell/\sqrt{\log n}}}$.
Applying a union bound (on $\ell$, $B$, $\tau$ and $Q$),
  $F$ does not occur with probability at most
$$
\Pr[F] \le \sum_{\ell\in [5n]} n^\ell\cdot 2^\ell\cdot 2^{\ell-1}\cdot (\phi\eps)^{\frac{a\ell}{\sqrt{\log n}}}  
\le \sum_{\ell\in [5n]} \left((4n)^{\frac{\sqrt{\log n}}{a}}\cdot \phi\eps\right)^{\frac{a\ell}{\sqrt{\log n}}}=o_n(1),
$$
where the factor $n^\ell2^\ell$ is an upper bound for the number of choices for $B$ of length $\ell$ and the initial configuration $\tau$ of $S(B)$, and the factor $2^{\ell-1}$ is because there can be no more than $\ell-1$  arcs in~a sequence of length $\ell$. The last equation follows by setting  $c_1$ in the choice of $\eps$ sufficiently large.

For the second step, first notice that when $F$ occurs, it follows
  from Lemma \ref{lem:main} that there exist no sequence $H$ of length $5n$ together
  with a configuration $\gamma$ of $S(H)$ so that 
  $H$ is  $\eps$-improving with
  respect to $\gamma$ and $X$.
Taking any implementation of the FLIP algorithm running on $G$ with weights $X$,
  this implies that the cut weight $\textsf{obj}$ must go up by at least $\eps$ for every $5n$ consecutive moves.
As the weight of any cut lies in $[-n^2,n^2]$, the number of 
  steps it takes to terminate is at most
$$
5n\cdot \frac{2n^2}{\eps} = \phi\cdot n^{O(\sqrt{\log n})}.
$$
This finishes the proof of Theorem \ref{maintheorem}.
\end{proof}

\section{Proof of the Main Lemma} 
\label{sec:sqrtn-ub}

We proceed now to the proof of Lemma \ref{lem:main}.
The plan of the section is as follows.
Given a nontrivial sequence $H=(\sigma_1,\ldots,\sigma_m)$ of moves of length $m=5n$, 
  we classify in Section \ref{arctypes} its arcs into good ones~and bad ones and 
  introduce the notion of the radius of an arc.
In Section \ref{sec:basiclemma} we prove a few basic lemmas that will be used in the proof
  of Lemma \ref{lem:main}.
Next we partition the set of all arcs of $H$~into chunks according to their lengths
  in Section \ref{sec:overview} and present an overview of cases of the proof of Lemma \ref{lem:main}.
There will be three cases, from Case $1$ to Case $3$, and they will be covered in
  Section \ref{sec:case1}, \ref{sec:case2} and \ref{sec:case3}, respectively.
For each case we choose $B$ to be either a substring or a subsequence of $H$.
Among all cases, there is only one occasion where we choose 
  $B$ to be a subsequence of $H$.~As discussed earlier in Section \ref{sec:maintheorem}, the second condition of 
  Lemma \ref{lem:main}
  is trivially satisfied when $B$ is a substring of $H$ (since we don't change the interior of any arc).
Therefore, there is no need to specify the configuration $\tau$
  in cases when $B$ is chosen to be a substring of $H$.

\subsection{Classification of arcs}\label{arctypes}

We start with a quick proof that there are many arcs in a long sequence of moves.

\begin{lemma}\label{lem:arc-length-ratio-lb}
For any sequence $B$ of moves, the number of arcs in $B$ is at least $\leng(B)-n$.
\end{lemma}
\begin{proof}
Denote by $\chi_v$ the number of occurrences of node $v$ in $B$. Then 
  the number of arcs in $B$ is $  \sum_{v\in V}(\chi_v - 1) \ge \sum_{v\in V}\chi_v - n = \leng(B) - n ~.$
\end{proof}

\begin{corollary}\label{cor:arc-length-ratio-lb-5n}
If $H$ is a sequence of moves of length $5n$, then it contains at least $4n$ arcs.
\end{corollary}

Next we give the definition of good and bad arcs in a sequence.
We start with some notation. Let $H=(\sigma_1,\ldots,\sigma_{m})$
  be a sequence of moves of length $m=5n$.
We use $\calA$ to denote the set of all arcs in $H$; by Corollary \ref{cor:arc-length-ratio-lb-5n}
  we have  $|\calA |\ge 4n$.

For each $k\in [m]$, we define the \emph{predecessor} $\pred_H(k)$ of
  the $k$-th move to be the largest index $i<k$ such that $\sigma_i=\sigma_k$ and set $\pred_H(k)$ to be $-\infty$ if no such index $i$ exists (i.e., $\pred_H(k)$ is the index of the previous occurrence of the same node $\sigma_k$).
Similarly we define the \emph{successor} $\succ_H(k)$ of the $k$-th move in $H$ to be the smallest index $j>k$ such that $\sigma_j=\sigma_k$ and set $\succ_H(k)$ to be $+\infty$ if no such $j$ exists. 
Next we define the radius of a move and an arc:


\begin{definition}[Radius]
For each $k\in [m]$ we define the \emph{radius}
  of the~$k$-th move in $H$ as $$\Radius{k}=\max\Big\{\leftradius{k},\rightradius{k}\Big\},$$  where
$$
\leftradius{k}=k-\pred_H(k)+1\quad\ \text{and}\ \quad
\rightradius{k}=\succ_H(k)-k+1.
$$


Given an arc $\alpha=(i,j)\in \calA$ of $H$,
  we define its radius as\footnote{Note that this is well defined because when $H$ is nontrivial, the $\interior(\alpha)$ of 
  every arc $\alpha$ is nonempty.}
$$\Radius{\alpha}= \max \Big\{ \Radius{k} \hspace{0.08cm} \Big| \hspace{0.08cm}  k \in \interior(\alpha) \Big\}.$$
\end{definition}

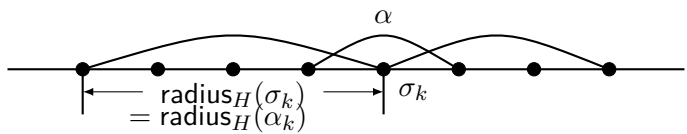
\begin{figure}[H]
\begin{tikzpicture}
\draw[fill=black] (-4,0) circle (2.5pt);
\draw[fill=black] (-3,0) circle (2.5pt);
\draw[fill=black] (-2,0) circle (2.5pt);
\draw[fill=black] (-1,0) circle (2.5pt);
\draw[fill=black] (0,0)  circle (2.5pt);
\draw[fill=black] (1,0)  circle (2.5pt);
\draw[fill=black] (2,0)  circle (2.5pt);
\draw[fill=black] (3,0) circle (2.5pt);

\node at (0,0.7) {$\alpha$};
\node at (0.4,-0.3) {$\sigma_k$};
\node at (-2,-0.35) {$\Radius{\sigma_k}$};
\node at (-2.2,-0.65) {$=\Radius{\alpha_k}$};
\draw[thick] (-5,0) -- (4,0);
\draw[thick] (-1,0) .. controls (-0,0.6) .. (1,0);
\draw[thick] (0,0) .. controls (-2,0.6) .. (-4,0);
\draw[thick] (0,0) .. controls (1.5,0.6) .. (3,0);
\draw[thick] (0,0) -- (0,-0.6);
\draw[thick] (-4,0) -- (-4,-0.6);
\draw[-latex] (-3.2,-0.3) -- (-4,-0.3);
\draw[-latex] (-0.8,-0.3) -- (0,-0.3);
\end{tikzpicture}
\centering
\caption{An example of radius}\label{exampleradius}
\end{figure}

It follows from the definition that if $\Radius{k}$ is not $+\infty$ 
  then both $\pred_H(k)$ and $\succ_H(k)$~are defined and then both
  $(\pred_H(k),k)$ and $(k,\succ_H(k))$ are arcs of $H$.
As an example of the radius of an arc $\alpha$, if there is a $k\in [i+1:j-1]$ such that 
  node $\sigma_k$ is adjacent to $\vertex_H(\alpha)$ in $G$ and 
  $\sigma_k$ does not appear anywhere else in $H$ (i.e., $\sigma_k\in S_1(H)$)  then $\Radius{\alpha}=+\infty$.
Another example is shown in Figure \ref{exampleradius}. Here $\Radius{\alpha}=\Radius{k}=5$
  assuming that $(\sigma_k,\sigma_i)$ is in $G$. 

Finally we define good arcs and bad arcs.

\begin{definition}[Good and bad arcs]\label{def:good}
We say an arc $\alpha=(i,j)$ of $H$ is \emph{good} if 
    \[
     \answer{\min\Big\{ \leftradius{i},\hspace{0.05cm} \rightradius{j}\Big\} \ge \frac{\leng(\alpha)}{2^{\lceil \sqrt{\log n}\rceil}} ~.}
    \]
    Otherwise we say that $\alpha$ is \emph{bad}.
Given a set of arcs $C\subseteq \mathcal{A}$ we write $\good{C}$ to denote the 
  set of good arcs in $C$ and $\bad{C}$ to denote the set of bad arcs in $C$.
\end{definition} 

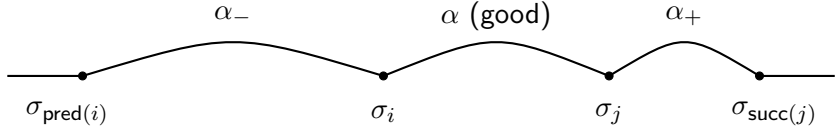
\begin{figure}[H]
\begin{tikzpicture}
\tikzstyle{every node}=[node font={\sf}]
\draw[fill=black] (-4,0) circle (1.5pt);
\draw[fill=black] (0,0)  circle (1.5pt);
\draw[fill=black] (3,0) circle (1.5pt);
\draw[fill=black] (5,0) circle (1.5pt);

\node at (-2,0.8) {$\alpha_-$};
\node at (1.5,0.8) {$\alpha$ (good)};
\node at (4,0.8) {$\alpha_+$};
\node at (0,-0.5) {$\sigma_i$};
\node at (3,-0.5) {$\sigma_j$};
\node at (-4.2,-0.5) {$\sigma_{\pred(i)}$};
\node at (5.2,-0.5) {$\sigma_{\succ(j)}$};
\draw[thick] (-5,0) -- (-4,0);
\draw[thick] (5,0) -- (6,0);
\draw[thick] (0,0) .. controls (-2,0.6) .. (-4,0);
\draw[thick] (0,0) .. controls (1.5,0.6) .. (3,0);
\draw[thick] (3,0) .. controls (4,0.6) .. (5,0);
\end{tikzpicture}
\centering
\caption{The arc $\alpha$ is a good arc iff for both sides of $\alpha$ the adjacent arcs have large length, i.e., both $\leng(\alpha_-)=\leng(\pred({i}),i)$ and $\leng(\alpha_+)=\leng(j,\succ(j))$ are at least $ {\leng(\alpha)}/{2^{\lceil \sqrt{\log n}\rceil}}$}\label{examplegood}
\end{figure}

Given a nonempty set of arcs $C\subseteq \calA$,
we write
$\Upper{C} = \max_{\alpha\in C} \{\leng(\alpha) \}$ and  
$$  \Appearances(C) = \Big\{k\in [m]: k = \leftt(\alpha) \text{ or } k = \rightt(\alpha) \text{ for some } \alpha \in C \Big\}.
$$

\subsection{Basic lemmas}\label{sec:basiclemma}

We start with a lemma that will be used to bound the
  rank of a set of arcs in a sequence.
It~is essentially the same as Lemma 3.2 in \cite{etscheid2015smoothed}, which connects the rank
  of a set $C$ of arcs with the number of distinct nodes
  that appear as endpoints of arcs in $C$.
We include a proof for completeness.
\begin{lemma}\label{firstlowerbound}
Let $C$ be a set of arcs of a nontrivial sequence $H$ such that 
  nodes of arcs in $C$ are all distinct.
Then we have 
$ \rank_H(C)\ge |C|/2 .$ 
\end{lemma}
\begin{proof} For every arc $\alpha \in C$, let
$v_{\alpha}$ be an arbitrary node in $\interior(\alpha)$
(it is nonempty since $H$ is nontrivial).
Construct a subset $C' \subseteq C$ of arcs as follows.
Start with $C' = \emptyset$. Pick an arbitrary arc $\alpha \in C$, add $\alpha$ to $C'$, remove from $C$ any arc whose node is $v_{\alpha}$ (if there is such
an arc), and repeat until $C$ becomes empty.
Since arcs of $C$ have distinct nodes,
adding an arc to $C'$ causes at most one other arc to be deleted from $C$. Hence, $|C'| \geq |C|/2$.

To see that the improvement vectors of arcs in $C'$ are 
  linearly independent, note for each $\alpha\in C'$
  that the $(\vertex_H(\alpha),v_\alpha)$-entry of $\improvement{H}{\alpha}$ is nonzero but
  the same entry of every other improvement vector 
  from $C'$ is $0$.
Therefore, the rank of $C'$ is full and
  $\rank(C)\ge \rank(C')\ge |C|/2$.
\end{proof}




We need some notation for the next lemma.
We say an arc $\alpha\in \calA$
  is \emph{contained} in an interval $I \subseteq [m]$ if both $\leftt(\alpha)$ and $\rightt(\alpha)$
  lie in $I$. For a set of arcs $C\subseteq \calA$ of $H$, 
 we write $\restrict{C}{I}$ to denote 
  the set of $\alpha\in C$ contained in $I$.
Intuitively $\restrict{C}{I}$ is the set of arcs in 
$C$
  that can be inherited by the substring $H_I$ of $H$.
The next lemma follows from (essentially) an averaging argument; the proof can be found in Appendix \ref{sec:app1}.

\begin{lemma}\label{lem:arc_node_ratio_in_cover}
Let $H$ be a sequence of moves of length $m$, $P\subseteq [m]$ be a nonempty
  set of indices and $C\subseteq \calA$ be a 
  nonempty set of arcs of $H$ such that $\leng(\alpha)\le \ell$ for all $\alpha\in C$
  for some positive integer parameter $\ell\le m/2$.
%
Then there exists an interval $I=[i:i+2\ell-1]\subseteq [m]$ of length $2\ell$ such that
    \begin{equation}\label{hehe2}
    \frac{\left| \restrict{C}{I} \right|}{|C|}  \ge \max\left\{    \frac{2\ell}{16m},\hspace{0.04cm}
       \frac{|P\cap I|}{4|P|} \right\}.
    \end{equation}
     Thus, it holds that 
    \begin{equation} \label{transfer-lower-bounds}
        \frac{\left| \restrict{C}{I} \right|}{\leng(I)}\ge \Omega\left(\frac{|C|}{\leng(H)}\right) \quad \text{ and } \quad \frac{\left| \restrict{C}{I} \right|}{|P\cap I|}\ge \Omega\left(\frac{|C|}{|P |}\right).
    \end{equation}
\end{lemma}

\subsection{Overview of cases}\label{sec:overview}

We now begin to prove Lemma \ref{lem:main}.
Let $G=(V,E)$ be a graph over $n$ vertices, $H=(\sigma_1,\ldots,\sigma_m)$
  be a nontrivial sequence of moves of length $m=5n$, and $\gamma$ be a configuration of $S(H)$.
Let $\calA$ be the set of all arcs of $H$ (with $|\calA|\ge 4n$ by Corollary \ref{coro1}).

We first partition $\calA$ into $s=\lceil \log m\rceil=\Theta(\log n)$ \emph{chunks}
  $\chunk_1,\ldots,\chunk_s$ according to lengths:\vspace{0.03cm}
    \[\chunk_i=\big\{\alpha\in \arcs:\leng(\alpha)\in [2^{i-1}+1:2^i] \big\}.\vspace{0.03cm} \]
Letting 
  $w=\lceil \sqrt{\log n}\rceil$ and~$t=$ $\left\lceil s/w\right\rceil$ so both $w$ and $t$ are $\Theta(\sqrt{\log n}),$
  next we assign these chunks to $t$ \emph{groups} $\group_1,\ldots,\group_t$:
$
\group_i=\chunk_{(i-1)w+1}\cup  \cdots \cup \chunk_{iw},
$
where the last group $\group_t$ may contain less than $w$ chunks.
It is easy to check from the definition of \emph{chunks} and \emph{groups} the following fact:
\begin{fact}\label{fact:bounds-groups-chunks}
If $P$ is a set of arcs from the same chunk $\chunk_j$, then for any arc $\alpha \in P$, it holds that 
\[\leng(\alpha)\le \Upper{P}\le 2\cdot\leng(\alpha)\]
If $Q$ is a set of arcs from the same group $\group_i$, then for any arc $\beta \in Q$, it holds that 
\[\leng(\beta)\le \Upper{Q} \le 2^w\cdot \leng(\beta) \]
\end{fact}

For each  chunk $\chunk_j$ in $\group_i$, we further partition it into two sets 
  $\chunk_j=\mathcal{L}_j\cup \mathcal{S}_j$ based on the radius  (when $\group_i=\emptyset$, $\mathcal{L}_j$ and $\mathcal{S}_j$
  are trivially empty even though $\Upper{\group_i}$ below is not defined):
    \begin{align*}\mathcal{L}_j &= \big\{\alpha\in \chunk_j : \Radius{\alpha} > 2\cdot \Upper{\group_{i}}\big\}\quad\text{and}\\[0.8ex]
    \mathcal{S}_j &= \big\{\alpha\in \chunk_j : \Radius{\alpha} \le 2\cdot \Upper{\group_{i}}\big\}.
    \end{align*}
Here $\mathcal{L}_j$ (the \emph{long-radius} arcs in $\chunk_j$) contains all arcs $\alpha\in \chunk_j$ such that there exists a $k\in \interior(\alpha)$ (where $\sigma_k$ is adjacent to $\vertex_H(\alpha)$ in $G$ and $\sigma_k$ 
     occurs an odd number of times inside $\alpha$) 
     such that $\Radius{k}$ is larger\footnote{Remember that this could be the case when for example either $\pred_H(k)=-\infty$  or $\succ_H(k)=+\infty$.} than $ 2\cdot \Upper{\group_{i}}$. The set $\mathcal{S}_j$ (the \emph{short-radius} arcs in $\chunk_j$) contains all arcs  $\alpha\in \chunk_j$ such that every $k\in \interior(\alpha)$ has predecessor and successor within $2\cdot \Upper{\group_{i}}$.

\begin{figure}[ h!]
\begin{tikzpicture}
\draw[thick] (-2.5,-1) -- (-2.5,1);
\draw[thick] (-7.5,-1) -- (-7.5,1);
\draw[thick] (-8.5,0) -- (-1.5,0);
\draw[thick, draw=WildStrawberry] (-6,0) .. controls (-5,0.5) .. (-4,0);
\draw[thick, draw=RoyalBlue, dashed] (-3,0) .. controls (-4,0.7) .. (-5,0);
\draw[thick, draw=RoyalBlue, dashed] (-8,0) .. controls (-6.5,0.7) .. (-5,0);

\draw[fill=RoyalBlue] (-8,0) circle (2.5pt);
\draw[fill=WildStrawberry] (-6,0) circle (2.5pt);
\draw[fill=WildStrawberry] (-4,0) circle (2.5pt);
\draw[fill=RoyalBlue] (-3,0) circle (2.5pt);
\draw[fill=RoyalBlue] (-5,0) circle (2.5pt);
\node at (-2.8,1.5) {$2\Upper{\group_i}$};

\draw[thick] (8-2.5,-1) -- (8-2.5,1);
\draw[thick] (8-7.5,-1) -- (8-7.5,1);
\draw[thick] (8-8.5,0) -- (8-1.5,0);
\draw[thick, draw=WildStrawberry] (8-6,0) .. controls (8-5,0.5) .. (8-4,0);
\draw[thick, draw=RoyalBlue, dashed] (8-7.2,0) .. controls (8-6.1,0.7) .. (8-5,0);
\draw[thick, draw=RoyalBlue, dashed] (8-2.8,0) .. controls (8-3.9,0.7) .. (8-5,0);

\draw[fill=RoyalBlue] (8-7.2,0) circle (2.5pt);
\draw[fill=WildStrawberry] (8-6,0) circle (2.5pt);
\draw[fill=WildStrawberry] (8-4,0) circle (2.5pt);
\draw[fill=RoyalBlue] (8-2.8,0) circle (2.5pt);
\draw[fill=RoyalBlue] (8-5,0) circle (2.5pt);
\node at (8-2.8,1.5) {$2\Upper{\group_i}$};

\end{tikzpicture}
\vspace{0.25cm}
\centering
\caption{A simplified example of \longradius (left) and \shortradius (right) arcs with only one node in their interior.}
\end{figure}
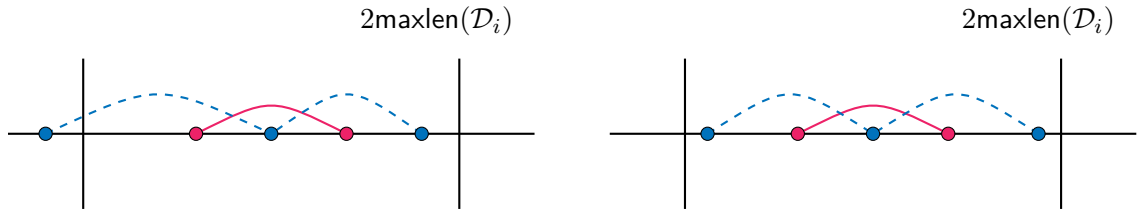

\clearpage 
Recall that our goal is to find a sequence $B$, a set $Q$ of arcs of $B$, and
  a configuration $\tau$ of $S(B)$ that satisfy both conditions of Lemma \ref{lem:main}.
The first case we handle in Section \ref{sec:case1} is when 
\begin{equation}\label{eq:case1}
{ \textbf{Case 1:}}\quad \big|\bad{\calA}\big| \ge 0.01\cdot |\calA|.
\end{equation}
Otherwise, $|\good{\calA}|\ge 0.99\cdot |\calA|$ and
  we pick a group $\group_{i^*}$ according to the following lemma:

\begin{lemma}\label{lem:existence-interval}
Assume that $|\good{\calA}|\ge 0.99\cdot |\calA|$. Then there exists a group $\group_{i^*}$,
  for some $i^*\in [t]$, that satisfies the following two conditions 
  (below we assume $\group_0=\group_{t+1}=\emptyset$ by default):
  \begin{equation*}  \frac{|\good{\group_{i^*}}|}{|\calA|} \ge \frac{1}{2t}\quad\ \text{and}\ \quad
    \frac{|\good{\group_{i^*}}|}{|\group_{i^*-1} \cup \group_{i^*}\cup \group_{i^*+1}|} \ge \frac{1}{7}.
    \end{equation*}
\end{lemma}
\begin{proof}
    For $\group_i$'s that violate the first inequality, the sum of their $|\good{\group_i}|$ is at most 
$ 
t\cdot (|\calA|/2t)$  $\le  {|\calA|}/{2}.
$ 
    For $\group_i$'s that violate the second inequality,  
      the  sum of their $|\good{\group_i}|$ is at most
\[
\sum_{i\in [t]} \frac{1}{7}\cdot \big|\group_{i+1}\cup \group_{i}\cup \group_{i-1}\big|
\le \frac{3}{7}\sum_{i\in [t]} |\group_i|=\frac{3|\calA|}{7}.
    \] 
Thus, the sum of $|\good{D}|$ where $D$ violates at least one inequality is at most $(\tfrac{1}{2}+\tfrac{3}{7})\calA<0.93 |\calA| $.
Given $|\good{\calA}|\ge 0.99\cdot |\calA|$,  there must exist a group $\group_{i^*}$ that satisfies both conditions.
\end{proof}
Fixing a group $\group_{i^*}$ that satisfies Lemma \ref{lem:existence-interval}
  and letting $\mathcal{L}=\cup_{\chunk_j\subseteq \group_{i^*}}\mathcal{L}_j$ and $\mathcal{S}=\cup_{\chunk_j\subseteq \group_{i^*}}\mathcal{S}_j$,
  we next split into two cases, depending on whether the majority of good arcs in $\group_{i^*}$ 
  are long-radius or short-radius.
We handle the second case in Section \ref{sec:case2} when
\begin{equation}\label{hehe10}
{ \textbf{Case 2:}}\quad
\big|\hspace{0.03cm}\good{\mathcal{L}} \big|\ge  0.5\cdot \big|\hspace{0.03cm}\good{\group_{i^*}}\big|
\end{equation}
and we handle the third case in Section \ref{sec:case3} when
\begin{equation}
{ \textbf{Case 3:}}\quad
\label{finalcase}\big|\hspace{0.03cm}\good{\mathcal{S}}\big|\ge  0.5\cdot \big|\hspace{0.03cm}\good{\group_{i^*}}\big|.
\end{equation}

\subsection{Case 1: Many arcs in $\calA$ are bad}\label{sec:case1}

\answer{We say an arc $\alpha=(i,j)$ of $H$ is \emph{dual-bad} if either $\pred_H(i)$ exists and $\leftradius{i}>\leng(\alpha)\cdot 2^w$, or 
 $\succ_H(j)$ exists and $\rightradius{j}>\leng(\alpha)\cdot 2^w$ where we recall that $w=\lceil \sqrt{\log n}\rceil$.
 Given a set of arcs $C\subseteq \mathcal{A}$, we write 
  $\dualbad{C}$ to denote the set of dual-bad arcs in $C$.}

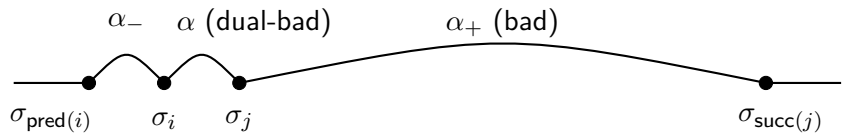
\begin{figure}[H]
\begin{tikzpicture}
\tikzstyle{every node}=[node font={\sf}]

\draw[fill=black] (-4,0) circle (2.5pt);
\draw[fill=black] (-3,0) circle (2.5pt);
\draw[fill=black] (-2,0) circle (2.5pt);
\draw[fill=black] (5,0) circle (2.5pt);

\node at (-3.5,0.8) {$\alpha_-$};
\node at (-1.8,0.8) {$\alpha$ (dual-bad)};
\node at (1.5,0.8) {$\alpha_+$ (bad)};
\node at (-3,-0.5) {$\sigma_i$};
\node at (-2,-0.5) {$\sigma_j$};
\node at (-4.5,-0.5) {$\sigma_{\pred(i)}$};
\node at (5.2,-0.5) {$\sigma_{\succ(j)}$};
\draw[thick] (-5,0) -- (-4,0);
\draw[thick] (5,0) -- (6,0);
\draw[thick] (-4,0) .. controls (-3.5,0.5) .. (-3,0);
\draw[thick] (-3,0) .. controls (-2.5,0.5) .. (-2,0);
\draw[thick] (-2,0) .. controls (1.5,0.7) .. (5,0);
\end{tikzpicture}
\centering
\caption{The arc $\alpha$ is dual-bad if either the adjacent $\leng(\alpha_-)$ or $\leng(\alpha_+)$ is longer than $\leng(\alpha)\cdot 2^w$. We can also interpret it as either $\alpha_-$ or $\alpha_+$ exists and is a bad arc.}\label{exampledualbad}
\vspace{-0.2cm}
\end{figure}
\answer{We prove the following lemma which implies Lemma~\ref{lem:main} for Case 1, by setting the subsequence $B$ to be $H_I$ and 
$Q$ to be the set of arcs of $B$ induced by $C$, with
$C$ and $I$ from the below statement.}
\begin{lemma}[\answer{Lemma~\ref{lem:main} for Case 1}]\label{lem:good-arcs-lb}
Assume that $|\bad{\calA}\big| \ge 0.01 \cdot |\calA|$.  
Then there exist an interval $I\subseteq [m]$ and a set~$C$ of arcs  of $H$ contained in $I$ such that $$\frac{\rank_H(C)}{\leng(I)}\ge \Omega\left(\frac{1}{\sqrt{\log n}}\right).$$
\end{lemma}
\begin{proof}
It follows from the definition of bad arcs that if $\alpha=(i,j)$ is bad, then either its previous arc $(\pred_H(i),i)$ or its next arc
  $(j,\succ_H(j))$ is dual-bad. As a result, the number of dual-bad arcs in $\calA$ is $\Omega(|\calA|)$ as well.
Thus, there exists a group $\group_{k}$ for some $k\in [t]$ (recall $t=\Theta(\sqrt{\log n})$) such that
    \[
\frac{|\dualbad{\group_{k}}|}{|\calA|} \ge \Omega\left(\frac{1}{\sqrt{\log n}}\right).
    \]
\answer{Let $\ell=\Upper{\dualbad{\group_{k}}}$. 
By the definition of dual-bad arcs, 
  $m\ge \ell\cdot 2^w$ and thus, $\ell\le m/2$.
It follows from Lemma \ref{lem:arc_node_ratio_in_cover} (setting $P=[m]$) that there is an interval $I$  of length $2\ell$ such that 
$$
\frac{|\restrict{\dualbad{\group_{k}}}{I}|}{\leng(I)}\ge \Omega\left(\frac{|\dualbad{\group_k}|}{\leng(H)}\right)
=\Omega\left(\frac{1}{\sqrt{\log n}}\right).
$$
where the last equation holds because $\leng(H)=m=\Theta(|\calA|)$.}

Let $C=\restrict{\dualbad{\group_k}}{I}$.
We  show that the rank of $C$ is almost full:
\begin{equation}\label{hehe6}
\rank_H(C)=\Omega(|C|)=\Omega(\ell/\sqrt{\log n}).
\end{equation}
The lemma then follows. 
To prove (\ref{hehe6}) we 
  show that for each $v\in V$, the number of~arcs~$\alpha\in C$ with 
  $\vertex_H(\alpha)=v$ can be bounded from above by a constant.
To see this is indeed the case, each $\alpha\in C $ with $\vertex_H(\alpha)=v$ must have one of its adjacent arcs $\beta$ (i.e.  
  either arc $(\pred_H(i),i)$ or $(j,\succ_H(j))$)
  defined, $\beta$ has at least one of its endpoints in $I$, and has length at least
$ 
\leng(\alpha)\cdot 2^{w}\ge \ell$ since $\alpha\in \group_k$ (and using Fact \ref{fact:bounds-groups-chunks}).
Given that the length of $I$ is $2\ell$, the number of such arcs $\beta$ (with the same $\vertex_H(\beta)=v$) is  
$O(1)$\footnote{In fact, the constant here is actually exactly one, since one dual-bad and two bad arcs of the same node has total length at least $3+2(\ell-1)=2\ell+1$ exceeding the $2\ell$ length of an interval.}.
Since each $\beta$ 
  can be adjacent to at most two arcs $\alpha\in C$,
  we conclude that the number of  $\alpha\in C$
  with $\vertex_H(\alpha)=v$ is $O(1)$.
As a result, there is a subset $C'$ of $C$ such that $|C'|=\Omega(|C|)$
  and all arcs in $C'$ have distinct nodes.
The inequality (\ref{hehe6}) then follows from 
$$
\rank_H(C)\ge \rank_H(C')\ge \Omega(|C'|)\ge \Omega(|C|)
$$
where the second inequality follows from Lemma \ref{firstlowerbound}.
\end{proof}

\subsection{Case 2: Most arcs in $\good{\group_{i^*}}$ are long-radius}\label{sec:case2}

We start with some intuition for the second case.

{In this case we work with a group 
  $\group_{i^*}$ that has many long-radius\footnote{Notice that the assumption of Case 2 is actually stronger, that there are many arcs in $\group_{i^*}$ that are both good and long-radius.
It turns out that we will only use their long-radius property in the proof of this case.} arcs. 
Recall that $\mathcal{L}$ is~the~set of long-radius arcs in $\group_{i^*}$.
Let $\ell=\Upper{\mathcal{L}}$.
It follows from Lemma \ref{lem:existence-interval} that there is an interval $I$ of length $2\ell$ such that $\restrict{\mathcal{L}}{I}$ is large. 
Furthermore, given that arcs in $\restrict{\mathcal{L}}{I}$
  are long-radius, one can show that every 
  $\alpha\in \restrict{\mathcal{L}}{I}$ has a 
  $k\in \interior(\alpha)$ such that either $\pred_H(k)$
  or $\succ_H(k)$ is outside of $I$.

Now if it is the  oversimplified scenario that
  \emph{both} $\pred_H(k)$ and $\succ_H(k)$ are outside of $I$,
  then one can conclude that $\restrict{\mathcal{L}}{I}$
  has full rank since only the improvement vector
  of $\alpha$ has a nonzero entry indexed by $(\vertex_H(\alpha),\sigma_k)$.
To see this is the case, note that there is no 
  arc with node $k$ contained in $I$ and no other arc
  with node the same as $\alpha$ can have $\sigma_k$ in
  its interior.
Our goal is to show for the
  general scenario that $\restrict{\mathcal{L}}{I}$
  has almost full rank. 

The following lemma 
  implies Lemma~\ref{lem:main} for the second case, by setting the subsequence $B$ to be $H_I$ and the
  set $Q$ to be arcs of $B$ induced by $C$,
  using $I$ and $C$ of the below statement.}
\begin{lemma}[\answer{Lemma~\ref{lem:main} for Case 2}]\label{secondlowerbound:2}
Assume that there exists a group $\group_{i^*}$ that satisfies Lemma \ref{lem:existence-interval} and (\ref{hehe10}).
Then there exist an interval $I\subseteq [m]$ and a set~$C$ of arcs  of $H$ contained in $I$ such that $$\frac{\rank_H(C)}{\leng(I)}\ge \Omega\left(\frac{1}{\sqrt{\log n}}\right).$$
\end{lemma}

\begin{proof}
Recall that we are in the case where we have a group $\group_{i^*}$ that satisfies 
  Lemma \ref{lem:existence-interval} and (\ref{hehe10}).
Given that $\mathcal{L} =\cup_{\chunk_j\subseteq \group_{i^*}} \mathcal{L}_j$, they together imply that 
$$
\left| \mathcal{L} \right|\ge \big|\hspace{0.03cm}\good{\mathcal{L}} \big|\ge  0.5\cdot \big|\hspace{0.03cm}\good{\group_{i^*}}\big|=\Omega\left(\frac{|\calA|}{\sqrt{\log n}}\right).
$$
Let $\ell=\Upper{\mathcal{L}}$. If $\ell\le m/2$
then, by Lemma~\ref{lem:arc_node_ratio_in_cover}, there is an interval $I\subseteq [m]$ of length
   $2\ell$ such that
$$
 \frac{|\restrict{\mathcal{L}}{I}|}{\leng(I)}=\Omega\left(\frac{|\mathcal{L} |}{m}\right)
=\Omega\left(\frac{1}{\sqrt{\log n}}\right).
$$
If $\ell > m/2$, then $\group_{i^*}$ is the last group ($i^*=t$). In this case, let $I=[m]$, and 
$\restrict{\mathcal{L}}{I} = \mathcal{L}$ satisfies the same property. In either case, let $C=\restrict{\mathcal{L} }{I}$.
We finish our proof by showing that $\rank_H(C)\ge |C|/2$.

It follows from the definition of  long-radius 
arcs that every $\alpha\in C$
  has a $k\in \interior(\alpha)$ such that $\Radius{k}>2\cdot \Upper{\group_{i^*}}
  \ge \leng(I)$.
For each arc $\alpha\in C$, we fix such a $k$ arbitrarily
  and~say~$\alpha$~is  \emph{left-long-radius} if $\leftradius{k} > \leng(I)$,
and \emph{right-long-radius} if $\rightradius{k} > \leng(I)$.

Without loss of generality we assume below that at least half of the 
  arcs in $C$ are left-long-radius (the argument is symmetric otherwise); we
  write $C'$ to denote the set of such arcs so $|C'|\ge |C|/2$.
We finish the proof of the lemma by showing that $\rank_H(C')$ is full.

To see this is the case, we order arcs in $C'$ by their right endpoints as
  $\alpha_1,\ldots,\alpha_{|C'|}$ such that 
  $$\rightt(\alpha_1)<\cdots<\rightt(\alpha_{|C'|}).$$
We write  $k_i$ to denote the index picked in $\interior(\alpha_i)$ such that 
  $\leftradius{k}> \leng(I)$~(this implies that 
  $\pred_H(k_i)$ must be outside of $I$). 

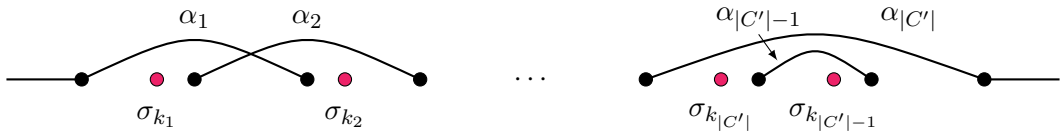
\begin{figure}[H]
\begin{tikzpicture}
\draw[fill=black] (-6,0) circle (2.5pt);
\draw[fill=black] (-4.5,0) circle (2.5pt);
\draw[fill=black] (-3,0) circle (2.5pt);
\draw[fill=black] (-1.5,0) circle (2.5pt);
\draw[fill=black] (1.5,0)  circle (2.5pt);
\draw[fill=black] (3,0)  circle (2.5pt);
\draw[fill=black] (4.5,0) circle (2.5pt);
\draw[fill=black] (6,0) circle (2.5pt);
\draw[fill=WildStrawberry] (-5,0) circle (2.5pt);
\draw[fill=WildStrawberry] (-2.5,0) circle (2.5pt);
\draw[fill=WildStrawberry] (2.5,0) circle (2.5pt);
\draw[fill=WildStrawberry] (4,0) circle (2.5pt);

\node at (-4.5,0.8) {$\alpha_1$};
\node at (-3,0.8) {$\alpha_2$};
\node at (3,0.8) {$\alpha_{|C'|-1}$};
\node at (5,0.8) {$\alpha_{|C'|}$};
\node at (0,0) {$\cdots$};
\node at (-5,-0.5) {$\sigma_{k_1}$};
\node at (-2.5,-0.5) {$\sigma_{k_2}$};
\node at (4,-0.5) {$\sigma_{k_{|C'|-1}}$};
\node at (2.5,-0.5) {$\sigma_{k_{|C'|}}$};
\draw[thick] (-6,0) -- (-7,0);
\draw[thick] (6,0) -- (7,0);
\draw[thick] (-6,0) .. controls (-4.5,0.7) .. (-3,0);
\draw[thick] (-4.5,0) .. controls (-3,0.7) .. (-1.5,0);
\draw[thick] (1.5,0) .. controls (3.75,0.8) .. (6,0);
\draw[thick] (3,0) .. controls (3.75,0.5) .. (4.5,0);
\draw[-latex] (3,0.6) -- (3.25,0.25);

\end{tikzpicture}
\centering
\caption{An illustration of Case 2.}\label{fig:case2}
\end{figure}
\clearpage 
We prove below that for each $i\in [|C'|]$, 
\begin{enumerate}
\renewcommand{\labelenumi}{(\theenumi)}
    \item the entry of  $\improvement{H}{\alpha_i}$ 
  indexed by  $(\vertex(\alpha_i), \sigma_{k_i} )$ is nonzero, and \label{cond:second-lb-1}
    \item the entry of
  $\improvement{H}{\alpha_j}$ indexed by the same edge is $0$ for all $j<i$. \label{cond:second-lb-2}
\end{enumerate}
It then follows from this triangular structure in the matrix formed by these
  vectors that $\rank_H(C')$ is full.
Note that (\ref{cond:second-lb-1}) follows trivially given that $k_i\in \interior(\alpha_i)$. So we focus on (\ref{cond:second-lb-2})
  in the rest of the proof.

Assume for a contradiction that the entry of $\improvement{H}{\alpha_j}$ indexed by
  $(\vertex(\alpha_i), \sigma_{k_i})$ is nonzero for some $j<i$.
  Then either $\vertex(\alpha_j)=\vertex(\alpha_i)$
  or $\vertex(\alpha_j)=\sigma_{k_i}) $.

\begin{flushleft}\begin{enumerate}
    \item[] \textbf{Case a:} $\vertex(\alpha_j)=\vertex(\alpha_i)$.
    Since the entry of $\improvement{H}{\alpha_j}$ indexed by
  $(\vertex(\alpha_i), \sigma_{k_i})$ is nonzero the node $\sigma_{k_i} $ must be in 
    $\interior (\alpha_j)$. However, $\alpha_j$ is on the left of $\alpha_i$ and they span disjoint intervals (because $\vertex(\alpha_j)=\vertex(\alpha_i)$).
    This contradicts with the fact that $\pred_H(k_i)$ is outside of $I$.
    \item[] \textbf{Case b:} $\vertex(\alpha_j)=\sigma_{k_i}$.
    Since the entry of $\improvement{H}{\alpha_j}$ indexed by
  $(\vertex(\alpha_i), \sigma_{k_i})$ is nonzero, $\vertex(\alpha_i) $ must be in 
    $\interior (\alpha_j)$.
    We know that  $\rightt(\alpha_j)<\rightt(\alpha_i)$ by the ordering of arcs.
    On the other hand, $\leftt(\alpha_j)\ge k_i$ since $\pred_H(k_i)$ is outside of $I$. As a result, $\alpha_j$ is contained
      in $\alpha_i$: $\leftt(\alpha_i)<\leftt(\alpha_j)
      <\rightt(\alpha_j)<\rightt(\alpha_i)$, and thus $\alpha_j$ cannot have $\vertex(\alpha_i)$ in its interior, a contradiction.
\end{enumerate}\end{flushleft}

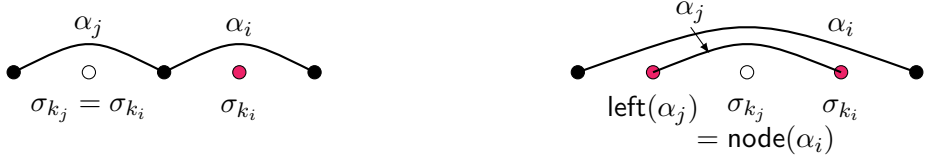
\begin{figure}[H]
\begin{tikzpicture}
\draw[fill=black] (-6,0) circle (2.5pt);
\draw[fill=black] (-4,0) circle (2.5pt);
\draw[fill=black] (-2,0) circle (2.5pt);
\draw[fill=WildStrawberry] (-3,0) circle (2.5pt);
\draw[fill=white] (-5,0) circle (2.5pt);

\draw[fill=black] (1.5,0)  circle (2.5pt);
\draw[fill=black] (6,0) circle (2.5pt);
\draw[fill=WildStrawberry] (2.5,0) circle (2.5pt);
\draw[fill=WildStrawberry] (5,0) circle (2.5pt);
\draw[fill=white] (3.75,0) circle (2.5pt);

\node at (-5,0.6) {$\alpha_j$};
\node at (-3,0.6) {$\alpha_i$};
\node at (-3,-0.5) {$\sigma_{k_i}$};
\node at (-5,-0.5) {$\sigma_{k_j}=\sigma_{k_i}$};

\node at (3,0.8) {$\alpha_{j}$};
\node at (5,0.6) {$\alpha_{i}$};
\node at (2.5,-0.5) {$\leftt(\alpha_j)$};
\node at (5,-0.5) {$\sigma_{k_i}$};
\node at (3.75,-0.5) {$\sigma_{k_j}$};
\node at (4,-0.9) {$=\vertex(\alpha_i)$};
\draw[thick] (-6,0) .. controls (-5,0.5) .. (-4,0);
\draw[thick] (-4,0) .. controls (-3,0.5) .. (-2,0);
\draw[thick] (1.5,0) .. controls (3.75,0.8) .. (6,0);
\draw[thick] (2.5,0) .. controls (3.75,0.5) .. (5,0);
\draw[-latex] (3,0.6) -- (3.25,0.25);

\end{tikzpicture}
\centering
\caption{An illustration of the proof in Case 2. Left: Case a; Right: Case b. The circles are the nodes where contradictions happen.}\label{fig:case2:contradiction}
\end{figure}

This finishes the proof of the lemma.\end{proof}
\subsection{Case 3: Most arcs in $\good{\group_{i^*}}$ are short-radius}\label{sec:case3}

Combining Lemma \ref{lem:existence-interval} and (\ref{finalcase}), we are left with the case that
\begin{equation}\label{hahahehebaba}
\frac{|\good{\mathcal{S}}|}{|\calA|} \ge 
    \Omega\left(\frac{1}{\sqrt{\log n}}\right) \quad\ \text{and}\ \quad
    \frac{|\good{\mathcal{S}}|}{|\group_{i^*-1} \cup \group_{i^*}\cup \group_{i^*+1}|} \ge \Omega(1)
\end{equation}
In Section \ref{sec:case300} we first  handle the easy case when 
  $i^*=t$ is the last group and its last two chunks 
  contain the majority of good arcs in $\mathcal{S}$:
\begin{equation}\label{hehehe1}
 {\big|\good{\mathcal{S}_{s-1}\cup\mathcal{S}_{s}}\big|}  \ge 
   0.5\cdot \big|\good{\mathcal{S}}\big|.
\end{equation}

\subsubsection{Case 3.0: $i^*=t$ and (\ref{hehehe1}) holds}\label{sec:case300}

\begin{lemma}[\answer{Lemma~\ref{lem:main} for Case 3.0}]\label{secondlowerbound:3.0}
Assume that $\group_{i^*}=\group_{t}$ and ${\big|\good{\mathcal{S}_{s-1}\cup\mathcal{S}_{s}}\big|}  \ge 
   0.5\cdot \big|\good{\mathcal{S}}\big|$.
Then there exist an interval $I\subseteq [m]$ and a set~$C$ of arcs  of $H$ contained in $I$ such that $$\frac{\rank_H(C)}{\leng(I)}\ge \Omega\left(\frac{1}{\sqrt{\log n}}\right).$$
\end{lemma}
\begin{proof}
On the one hand, from (\ref{hehehe1}) we have trivially
  that $|\chunk_{s-1}\cup \chunk_s|\ge \Omega(m/\sqrt{\log n})$.
On the other hand,
   every $\alpha\in \chunk_{s-1}\cup \chunk_s$ 
  has length $\Omega(m)$ and thus, the number of arcs 
  $\alpha\in \chunk_{s-1}\cup \chunk_s$ with the same $\vertex(\alpha)$ is bounded from above by a constant.
As a result, there is a subset of $\chunk_{s-1}\cup \chunk_s$ such that 
  its size is $\Omega(|\chunk_{s-1}\cup \chunk_s|)$ and all arcs in it have
  distinct nodes. It follows from Lemma \ref{firstlowerbound} that
$$
\rank_H(\chunk_{s-1}\cup \chunk_s)\ge \Omega(|\chunk_{s-1}\cup \chunk_s|)\ge   \Omega\big(m/\sqrt{\log n}\big).
$$
Lemma \ref{lem:main} in this case then follows by setting $B=H$ itself and $Q=\chunk_{s-1}\cup \chunk_s$.
\end{proof}

\subsubsection*{$\mathbf{\Touches}$ vs $\mathbf{\NonTouches}$}
Ruling out Case 3.0, we have from (\ref{hahahehebaba})
  that there is a $\chunk_{j^*}$ in $\group_{i^*}$ 
  with $j^*\le s-2$ such that
    \[
    \frac{|\good{\mathcal{S}_{j^*}}|}{|\calA|} \ge 
  \Omega\left(\frac{1}{ {\log n}}\right)\quad\ \text{and}\ \quad
    \frac{|\good{\mathcal{S}_{j^*}}|}{| \group_{i^*+1}  \cup \group_{i^*}\cup \group_{i^*-1}|} \ge \Omega\left(\frac{1}{\sqrt{\log n}}\right).
    \]
Setting $C^*=\good{\mathcal{S}_{j^*}}$, $\ell=\Upper{C^*}\le m/2$
  since $j^*\le s-2$, and $$P=\Appearances\big(\group_{i^*-1}\cup \group_{i^*}\cup \group_{i^*+1}\big)\subseteq [m]$$
  (so $|P|=\Theta(|\group_{i^*-1}\cup \group_{i^*}\cup \group_{i^*+1}|)$), by Lemma~\ref{lem:arc_node_ratio_in_cover}
 there is an interval $I$ of length  $2\ell$ such that

\begin{equation}\label{hahe1}
\frac{|\restrict{C^*}{I}|}{\leng(I)} = \Omega\left(\frac{1}{\log n}\right)\ \quad\text{and}\quad\ 
       \frac{|\restrict{C^*}{I}|}{|P\cap I|} = \Omega\left(\frac{1}{\sqrt{\log n}}\right).
\end{equation}
For convenience we write $C =\restrict{C^*}{I}$ in the rest of the proof. Every arc $\alpha\in C$ lies in $\mathcal{S}_{j^*}$, has
  length between $\ell/2$ and $\ell$, is good and \shortradius,
  and is contained in $I$ (which has length $2\ell$).

We need the following definition of two arcs overlapping with each other:

\begin{definition} 
We say that two arcs $\alpha=(i,j)$ and $\beta =(k,\ell)$ of $H$ \emph{overlap}  if $\vertex(\alpha)\ne \vertex(\beta)$ are adjacent in the graph and the endpoints of the arcs satisfy 
  $i < k < j < \ell$ or $k < i < \ell < j$.
\end{definition}

Given $C$, we say an arc $\beta=(i,j)$ is \emph{endpoint-disjoint}
  from $C $ if $i,j\notin \Appearances(C )$, i.e. 
  $\beta$ shares no endpoint with any arc $\alpha\in C $.
We write $\Touches$ to denote the set
  of arcs $\beta=(i,j)\in \calA$ 
  that are contained in $I$, are endpoint-disjoint from $C$, and 
  overlap with at least one arc in $ C$.
On the other hand, we write $\NonTouches$ to denote the set of arcs 
  $\beta\in \calA$ that are contained in $I$, endpoint-disjoint from $C$,
  and do not overlap with any arc in $C$.
  
We distinguish two cases depending on the
size of the set $\Touches$.
Section \ref{sec:case31} handles 
Case 3.1 when $ \Touches $ is "large", specifically,
\begin{equation}\label{eq:touches_large}
\frac{|\Touches|}{\leng(I)} \ge    \frac{\big|\restrict{(\group_{i^*+1}  \cup \group_{i^*}\cup \group_{i^*-1})}{I}\hspace{-0.06cm}\big|}{\leng(I)}+\frac{1}{\sqrt{\log n}}, \end{equation}
and Section \ref{sec:case32} handles the opposite case 3.2, when $ \Touches $ is "small".


\subsubsection{Case 3.1: $\Touches$ is large}\label{sec:case31}

\begin{lemma}[\answer{Lemma~\ref{lem:main} for Case 3.1}]\label{secondlowerbound:3.1}
Assume that condition (\ref{eq:touches_large}) holds, i.e $\Touches$ is large. Then for the interval $I\subseteq [m]$ of Equation~(\ref{hahe1}) there exists a set~$F$ of arcs of $H$ contained in $I$ such that $$\frac{\rank_H(F)}{\leng(I)}\ge \Omega\left(\frac{1}{\sqrt{\log n}}\right).$$
\end{lemma}
\begin{proof}
Let $F=\Touches \backslash (\restrict{(\group_{i^*+1}\cup \group_{i^*}\cup \group_{i^*-1})}{I}).$
In this case we have
\begin{equation}\label{eqn:case2-1-touches-lb}
\frac{|F|}{\leng(I)} \ge    \frac{1}{\sqrt{\log n}} .   
\end{equation}

For every $\alpha\in F$ we pick arbitrarily 
  an arc $\beta \in C$ such that $\alpha$ and $\beta$ overlap;
  we call $\beta$ a \emph{witness arc} for $\alpha$.
Assume without loss of generality that for the majority of $\alpha\in F$,
  its witness arc $\beta\in C $ is on the left of $\alpha$, meaning
  that $\leftt(\beta)<\leftt(\alpha)$; the argument is symmetric otherwise.
We write $F'$ to denote the subset of such arcs in $F$; we have
  $|F'|\ge |F|/2$.
  
\begin{figure}[H]
\begin{tikzpicture}
\draw[thick] (-5,0) -- (4,0);
\draw[thick,draw=RoyalBlue, dashed] (-4,0) .. controls(-1.5,0.6) .. (1,0);
\draw[thick,draw=RoyalBlue, dashed] (-2,0) .. controls(0.5,.6) .. (3,0);
\draw[thick,draw=WildStrawberry] (-3,0) .. controls(-2,-0.6) .. (-1,0);
\draw[thick,draw=WildStrawberry] (0,0) .. controls(1,-0.6) .. (2,0);
\draw[thick] (-1,0) .. controls(-0.5,-0.4) .. (0,0);

\draw[fill=RoyalBlue] (-4,0) circle (2.5pt);
\draw[fill=WildStrawberry] (-3,0) circle (2.5pt);
\draw[fill=RoyalBlue] (-2,0) circle (2.5pt);
\draw[fill=WildStrawberry] (-1,0) circle (2.5pt);
\draw[fill=WildStrawberry] (0,0) circle (2.5pt);
\draw[fill=RoyalBlue] (1,0) circle (2.5pt);
\draw[fill=WildStrawberry] (2,0) circle (2.5pt);
\draw[fill=RoyalBlue] (3,0) circle (2.5pt);

\node at (-5.5,0) {$I:=$};
\node at (-4.5,0) {$[$};
\node at (3.5,0) {$]$};
\node at (-1.5,0.8) {\textcolor{RoyalBlue}{$\beta_1$}};
\node at (0.5,0.8) {\textcolor{RoyalBlue}{$\beta_2$}};
\node at (-2,-0.8) {\textcolor{WildStrawberry}{$\alpha_2$}};
\node at (1,-0.8) {\textcolor{WildStrawberry}{$\alpha_1$}};

\end{tikzpicture}
\centering
\caption{An example of $\Touches$ and witness. Here $\vertex(\beta_1)$ and $\vertex(\beta_2)$ are adjacent to $\vertex(\alpha_1)=\vertex(\alpha_2)$ in $G$.
The dashed blue arcs are in $C$, the red arcs are in $\Touches$, and the black arc is in $\NonTouches$. $\beta_1$ is the witness of $\alpha_1\in F'$ (i.e. $\beta_1$ is on the left of $\alpha_1$), and $\beta_2$ is the witness of $\alpha_2\in F\backslash F'$.}\label{fig:witness}
\end{figure}
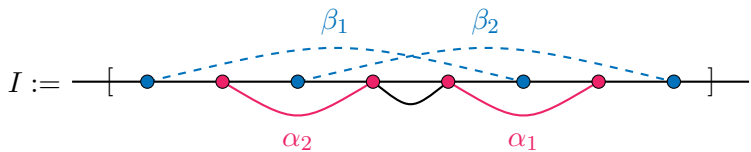

Next we partition the interval $I$ into five quantiles so that each one is of length 
  $\lfloor \leng(I)/5\rfloor$ or $\lceil\leng(I)/5\rceil$.~We also assume that $\leng(I)$ is sufficiently large so that 
  $\lceil\leng(I)/5\rceil<\leng(I)/4$; otherwise it is bounded by a constant and our goal is trivially met. (Note that 
  $C$ is nonempty. So $I$ is an interval of constant length and contains
  at least one arc. We can thus prove Lemma~\ref{lem:main} directly just by taking $B=H_I$ and one single arc in it;
  the ratio in (\ref{aslarge}) is $\Omega(1)$.)
Let $\beta$ be the witness arc of an $\alpha\in F'$.
Given that $ \beta\in C\subseteq \chunk_{j^*}$,
  we have $$\leng(\beta) > \frac{\Upper{\chunk_{j^*-1}}}{2} 
  \ge \frac{\ell}{2}= \frac{ \leng(I)}{4}>\left\lceil \frac{\leng(I)}{5}\right\rceil.$$
Since $\beta$ is contained in $I$, $\rightt(\beta)$ can not lie in the first quantile of $I$.
Partitioning $F'$ into $F'_i$ for $i=2,3,4,5$ so that $F'_i$
  contains all arcs $\alpha\in F'$ such that $\rightt(\beta)$ of its 
  witness arc $\beta$ lies in the $i$-th quantile of $I$ and 
  letting $F'_q$ denote the largest set, we have
$$
|F'_q|\ge \frac{|F'|}{4}\ge \frac{|F|}{8}=\Omega\left(\frac{\leng(I)}{
\sqrt{\log n}}\right).
$$

We finish the proof by showing that  $\rank_H(F'_q)$ is full.
The proof that $\rank_H(F'_q)$ is full is similar to the proof of Case 2.
We order arcs $\smash{\alpha_1,\ldots,\alpha_{|F'_q|}}$ in $F'_q$  by the 
  right endpoints of their witness arcs: $\beta_i$ is the witness arc of $\alpha_i$ and they satisfy $$\rightt(\beta_1)\le \cdots\le  {\rightt(\beta_{|F'_q|})}.$$
Note that we used $\le$ instead of $<$ because some of the
  witness arcs might be the same.
\clearpage   
We prove below that for each $i\in [|F'_q|]$, 
\begin{enumerate}
\renewcommand{\labelenumi}{(\theenumi)}
    \item the  entry of $\improvement{H}{\alpha_i}$ indexed by edge $(\vertex(\alpha_i),$ $\vertex(\beta_{ i}))$ is nonzero, and 
    \item the entry of
  $\improvement{H}{\alpha_j}$ indexed by the same edge is $0$ for every $j<i$.
\end{enumerate}
It follows from this triangular structure in the matrix that $\rank_H(F'_q)$ is full.
  
For (1) we first note that $(\vertex_H(\alpha_i),\vertex_H(\beta_i))$ is an edge in $G$ because $\alpha_i$ and $\beta_i$ overlap.
As $\beta_i\in C $ is a good arc,
  we have $$\rightradius{\beta_i}\ge \frac{\leng(\beta)}{2^w}.$$
On the other hand, by the definition of $F$, $\alpha_i$
  belongs to $\group_k$ for some $k<i^*-1$ (note that $k$ cannot be larger than 
  $i^*+1$ because  $\alpha_i$ would then be too long to be contained in $I$).
If $\succ_H(\rightt(\beta_i))$ lies inside $\alpha_i $,
  we have $\leng(\alpha)> {\leng(\beta)}/2^w$.
But this contradicts with the fact that $\beta\in \chunk_{j^*}\subseteq \group_{i^*}$
  but $\alpha\in \group_k$ for some $k<i^*-1$ and thus,
  $\leng(\alpha)/\leng(\beta)<1/2^w$.
As a result, $\succ_H(\rightt(\beta_i))>\rightt(\alpha_i)$ and thus,
  $\vertex(\beta_i)$ appears exactly once inside $\alpha_i$.
This implies (1).

\begin{figure}[H]
\centering
\begin{tikzpicture}
\draw[thick, draw=RoyalBlue] (-4,0) .. controls (-2,0.5) .. (0,0);
\draw[thick] (0,0) .. controls (1.5,0.5) .. (3,0);
\draw[thick, draw=WildStrawberry] (-1,0) .. controls (0,0.5) .. (1,0);

\draw[fill=RoyalBlue] (-4,0) circle (2.5pt);
\draw[fill=WildStrawberry] (-1,0) circle (2.5pt);
\draw[fill=RoyalBlue] (0,0) circle (2.5pt);
\draw[fill=WildStrawberry] (1,0) circle (2.5pt);
\draw[fill=RoyalBlue] (3,0) circle (2.5pt);

\node at (-2.2,0.8) {\textcolor{RoyalBlue}{$\beta_i\in \good{\mathcal{S}_{j^*}}$}};
\node at (0.5,0.8) {\textcolor{WildStrawberry}{$\alpha_i\in \group_{<i^*-1}$}};
\node at (3,-0.5) {$\sigma_{\succ(\rightt(\beta_i))}$};
\node at (1,-0.5) {$\sigma_{\rightt(\alpha_i)}$};

\end{tikzpicture}
\centering
\caption{An illustration of (1). The conditions guarantee that there is only one appearance of $\vertex(\beta_i)$ in $\interior(\alpha_i)$. }\label{fig:case31-cond1}
\end{figure}
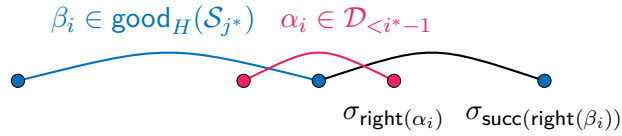

To show (2), assume for a contradiction that the entry of $\improvement{H}{\alpha_j}$
  indexed by the edge $(\vertex_H(\alpha_i),\vertex_H(\beta_i))$ is nonzero for some $j<i$.
  Then $\vertex_H(\alpha_j)$ is either $\vertex_H(\alpha_i)$ or $\vertex_H(\beta_i)$.
We consider the two cases:

\begin{flushleft}
\begin{minipage}{15cm}
\begin{enumerate}
    \item[] \textbf{Case a:} $\vertex_H(\alpha_j)=\vertex_H(\alpha_i)$.
    Since the entry of $\improvement{H}{\alpha_j}$
  indexed by $(\vertex_H(\alpha_i),\vertex_H(\beta_i))$ is nonzero, $\vertex_H(\beta_i)$ occurs in 
  $\interior(\alpha_j)$. 
    From $\rightt(\beta_j)\le \rightt(\beta_i)$
     we have $\rightt(\alpha_j)\le \leftt(\alpha_i)$,
     i.e. $\alpha_j$ is left of $\alpha_i$.
     Every arc of $F'_q$, and in particular $\alpha_j$, has length less than $\leng(I)/2^w$ and has in its interior a point in the $q$-th quantile (namely, $\rightt(\beta_j)$). Hence both endpoints of $\alpha_j$ are either inside the $q$-th quantile or at distance at most $\leng(I)/2^w$.
However, given that $\rightt(\beta_i)$ lies in the $q$-th quantile
      and that $\leng(\beta_i) \ge \leng(I)/4$,
      there cannot be any occurrence of $\vertex(\beta_i)$ to the left
      of $\rightt(\beta_i)$ inside, or within distance $\leng(I)/2^w$ of, the $q$-th quantile; this contradicts the assumption that $\vertex_H(\beta_i))$ occurs in  $\interior(\alpha_j)$.
    \item[] \textbf{Case b:} $\vertex_H(\alpha_j)=\vertex_H(\beta_i)$.
      Since the entry of $\improvement{H}{\alpha_j}$
  indexed by $(\vertex_H(\alpha_i),\vertex_H(\beta_i))$ is nonzero, $\vertex_H(\alpha_i))$ occurs in 
  $\interior(\alpha_j)$.
    First we note that $\rightt(\beta_j) \neq \rightt(\beta_i)$ because $\vertex_H(\beta_j)\ne\vertex_H(\alpha_j)=\vertex_H(\beta_i)$. Since $j<i$ we have therefore
      $\rightt(\beta_j)<\rightt(\beta_i)$.
    Next, given that $\rightt(\beta_j)$ is in the $q$-th quantile 
      and that $\leftt(\beta_i)$ is outside of the $q$-th quantile,
      we have that $\rightt(\beta_j)$ is contained in $\beta_i$.
    Given that it is also contained in $\alpha_j$ and 
      that $\vertex_H(\alpha_j)=\vertex_H(\beta_i)$,
we must have $\alpha_j=\beta_i$,
      a contradiction since $\beta_i\in  \group_{i^*}$ but
      $\alpha_j$ lies in $\group_k$ for some $k<i^*-1$.
\end{enumerate}
\end{minipage}
\end{flushleft}

\begin{figure}[H]
\begin{tikzpicture}
\draw[thick] (-2.5,-1) -- (-2.5,1);
\draw[thick] (-7.5,-1) -- (-7.5,1);
\draw[thick] (-8.5,0) -- (-1.5,0);
\draw[thick, draw=WildStrawberry] (-2,0) .. controls (-3,0.5) .. (-4,0);
\draw[thick, draw=WildStrawberry] (-4,0) .. controls (-5,0.5) .. (-6,0);
\draw[thick, draw=RoyalBlue, dashed] (-3,0) .. controls (-5.5,0.7) .. (-8,0);

\draw[thick] (5.8,-1) -- (5.8,1);
\draw[thick] (2,-1) -- (2,1);
\draw[thick] (0.5,0) -- (6.5,0);
\draw[thick, draw=RoyalBlue, dashed] (1,0) .. controls (3,0.7) .. (5,0);
\draw[thick, draw=RoyalBlue, dashed] (1.5,0) .. controls (3.5,0.7) .. (5.5,0);
\draw[thick, draw=WildStrawberry] (4,0) .. controls (5,-0.5) .. (6,0);
\draw[thick, draw=WildStrawberry, densely dotted] (1.5,0) .. controls (3.5,-0.7) .. (5.5,0);

\draw[fill=RoyalBlue] (-8,0) circle (2.5pt);
\draw[fill=WildStrawberry] (-6,0) circle (2.5pt);
\draw[fill=white] (-5,0) circle (2.5pt);
\draw[fill=WildStrawberry] (-4,0) circle (2.5pt);
\draw[fill=RoyalBlue] (-3,0) circle (2.5pt);
\draw[fill=WildStrawberry] (-2,0) circle (2.5pt);

\draw[fill=RoyalBlue] (1,0)  circle (2.5pt);
\draw[fill=RoyalBlue] (1.5,0)  circle (2.5pt);
\draw[fill=RoyalBlue] (5,0)  circle (2.5pt);
\draw[fill=RoyalBlue] (5.5,0)  circle (2.5pt);
\draw[fill=WildStrawberry] (4,0) circle (2.5pt);
\draw[fill=WildStrawberry] (6,0) circle (2.5pt);

\node at (-5,0.8) {\textcolor{WildStrawberry}{$\alpha_j$}};
\node at (-3,0.8) {\textcolor{WildStrawberry}{$\alpha_i$}};
\node at (-6.5,0.8) {\textcolor{RoyalBlue}{$\beta_i$}};
\node at (-2.8,1.5) {$F_q'$};
\node at (-5,-0.5) {$\vertex(\beta_i)$};

\node at (5.5,1.5) {$F_q'$};
\node at (3.5,-0.8) {\textcolor{WildStrawberry}{$\alpha_{j}$}};
\node at (5,-0.8) {\textcolor{WildStrawberry}{$\alpha_{i}$}};
\node at (2.75,0.8) {\textcolor{RoyalBlue}{$\beta_{j}$}};
\node at (4.25,0.8) {\textcolor{RoyalBlue}{$\beta_{i}$}};

\end{tikzpicture}
\vspace{0.15cm}
\centering
\caption{An illustration of the proof in Case 3.1. Left: Case a; Right: Case b. The circles and the dotted arc are the places where contradictions happen.}\label{fig:case31}
\vspace{-.5cm}
\end{figure}
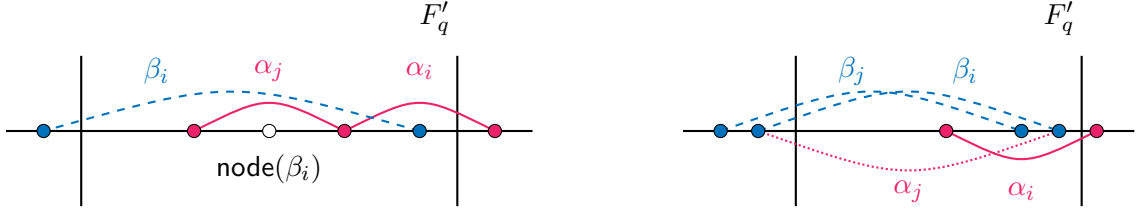
\end{proof}

Lemma \ref{lem:main} follows by taking $B=H_I$ and $Q$ to be the set of arcs of $B$ induced by the quantile $F'_q$ of $F$. 
\subsubsection{Case 3.2: $\Touches$ is small}\label{sec:case32}

We are in the last case when
$$
\frac{|\Touches|}{\leng(I)} < \frac{\big|\restrict{(\group_{i^*+1}  \cup \group_{i^*}\cup \group_{i^*-1})}{I}\hspace{-0.06cm}\big|}{\leng(I)}+ \frac{1}{\sqrt{\log n}}. $$
Given that $|\restrict{(\group_{i^*+1}  \cup \group_{i^*}\cup \group_{i^*-1})}{I}\hspace{-0.06cm}|
  \le |P\cap I|$, we have from (\ref{hahe1})
\begin{equation}\label{hehe100}
|\Touches|< |P\cap I|+\frac{|\leng(I)|}{\sqrt{\log n}} =O\left(\sqrt{\log n}\cdot |C|\right).
\end{equation}

Let $\gamma$ be the configuration of $S(H)$ in the statement of Lemma \ref{lem:main},
  and $\gamma'$ be the configuration of $S(H)$ before the first move of $I$.
We start with a sketch of the proof for this case.
First we note that we can apply Lemma \ref{firstlowerbound} to conclude that 
  $\rank_H(C )=\Omega(|C |)$.
This is again because the length of $I$ is $2\ell$,  
  all arcs in $C$ are contained in $I$, and have length at least $\ell/2$. Hence we can pick a subset $C'$ of $C$ of size $|C'| \geq |C|/4$ such that the arcs of $C'$ have distinct nodes.  Lemma \ref{firstlowerbound} then implies that $\rank_H(C ) \geq \rank_H(C' )=\Omega(|C |)$.
  
The main step of the proof is to construct a subset $R\subset I$ that satisfies 
  $R\cap \Appearances{(C)}=\emptyset$.
We remove the moves in $R$ to obtain the desired subsequence $B$ of $H$:
  $B=H_{I\setminus R}$, and let $\tau$ be the restriction of configuration $\gamma'$ on $S(B)$.
For each $i\in I\setminus R$, we use $\rho(i)\in [|B|]$ to denote its corresponding 
  index in $B$.
Then each arc $\alpha =(i ,j )\in C $ corresponds to an arc $\rho(\alpha )=(\rho(i ),\rho(j))$
  in $B$ (since $R\cap \Appearances(C)=\emptyset$, 
  both $i$ and $j$ survive), and we write $Q$ to denote the set of $|C|$ arcs of $B$ that consists of
  $\rho(\alpha )$, $\alpha \in C $.
  
The key property we need from the set $R$ is that the removal of moves in $R$
  does not change the improvement vector of any arc $\alpha \in C $.
More formally, we prove in Lemma \ref{lem:nochange} that
  $\improvement{\gamma,H}{\alpha }=\improvement{\tau,B}{\rho(\alpha)}$
  for all $\alpha \in C $.
It follows that (1) $B,Q$ and $\tau$ satisfy the second condition of Lemma \ref{lem:main}, and
  (2) $\textsf{rank}_B(Q)=\textsf{rank}_H(C )=\Omega(|C |)$.
To finish the proof of Lemma \ref{lem:main}, we prove in Lemma \ref{lem:largeratio}
  that the length of $B$ is small: $\smash{|B|= |I\setminus R|\le O(\sqrt{\log n})\cdot |C|}$.


We now construct $R$. 
To help the analysis in Lemma \ref{lem:nochange} we will consider $R$ as being composed of three parts, $R=R_1\cup R_2\cup R_3$. 
Given the plan above  we would like to add as many indices $i\in I\setminus \Appearances({C})$
  to $R$ as possible since the smaller $I\setminus R$ is, the larger the
  ratio $|C |/|I\setminus R|$ becomes. At the same time we need to maintain the key property
  that the removal of $R$ does not change the improvement vector of any arc $\alpha\in C$.


For each node $u\in S(H_I)$, we consider the following cases: (\textbf{a}) $u\in S_1(H_I)$, i.e the node $u$ appears exactly once in the interval,
(\textbf{b}) $u\in S_2(H_I)$ and $u$ appears an even number of times and (\textbf{c}) $u\in S_2(H_I)$  and $u$ appears an odd number of times.

\begin{flushleft}
\begin{enumerate}
    \item[] 
    \begin{minipage}{15cm}
    \textbf{Case a:} $u\in S_1(H_I)$.
    Let $k\in I$ with $\sigma_k=u$ be the unique occurrence of $u$ in $I$.
    If the radius of $k$ is \longradius: $\Radius{k}>2\cdot \Upper{\group_{i^*}}$,
     we add $k$ to $R_1$; we leave $k$ in $I\setminus R$ otherwise.
    The idea here is that if the radius of $k$ is \longradius, then given that every
      arc $\alpha \in C $ is \shortradius, we have $k\notin \interior(\alpha )$ and thus,
      the removal of $k$ does not affect the improvement vector of $\alpha $.
    On the other hand, if the radius of $k$ is small, then it is an endpoint of two
      arcs $(\pred_H(k),k)$ and $(k,\succ_H(k))$ 
    and both have length at most $2\cdot \Upper{\group_{i^*}}$. At the same time, 
      given that $u$ only appears once in $I$ and that $I$ has length $2\ell$, at least one of them has   length at least $\ell\ge 2^{j^*-1}+1$.
    As a result, we have
      $k\in P$ when $k$ is not added to $R_1$.
    \end{minipage}

    \item[] 
    \begin{minipage}{15cm}
    \textbf{Case b:} $u\in S_2(H_I)$ and $u$
      appears an even number of times in $H_I$.
    Let $k_1<k_2<\cdots<k_{2q}$ be the occurrences of $u$ in $I$
      for some $q\ge 1$.
    Then for each $i\in [q]$, we add both $k_{2i-1}$ and $k_{2i}$ to $R_2$ if 
      $(k_{2i-1},k_{2i})\in \NonTouches$, and keep both of them
      in $I\setminus R$ otherwise.
    Note that if $(k_{2i-1},k_{2i})\notin \NonTouches$, then  either $(k_{2i-1},k_{2i})\in \Touches$ or at least 
        one of the two endpoints is in $\Appearances(C)$.
    As a result, we can conclude that the number of these $2q$ indices that 
      do not get added to $R_2$ can be bounded from above by
      $$
      O\Big(\text{number of $\beta\in \Touches\cup C $ with $\vertex_H({\beta})=u$}\Big).
      $$
    \end{minipage} 
    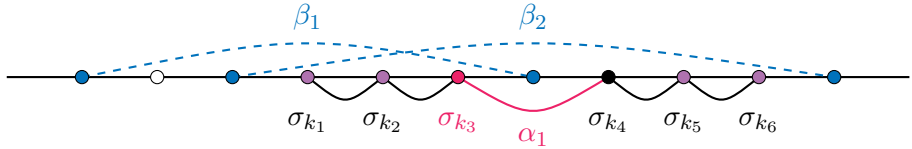
\begin{figure}[H]
    \begin{tikzpicture}
    \draw[thick] (-5,0) -- (7,0);
    \draw[thick,draw=RoyalBlue, dashed] (-4,0) .. controls(-1,0.6) .. (2,0);
    \draw[thick,draw=RoyalBlue, dashed] (-2,0) .. controls(2,0.6) .. (6,0);
    \draw[thick,draw=WildStrawberry] (1,0) .. controls(2,-0.6) .. (3,0);
    \draw[thick] (-1,0) .. controls(-0.5,-0.4) .. (0,0);
    \draw[thick] (0,0) .. controls(0.5,-0.4) .. (1,0);
    \draw[thick] (3,0) .. controls(3.5,-0.4) .. (4,0);
    \draw[thick] (4,0) .. controls(4.5,-0.4) .. (5,0);
    
    \draw[fill=RoyalBlue] (-4,0) circle (2.5pt);
    \draw[fill=white] (-3,0) circle (2.5pt);
    \draw[fill=RoyalBlue] (-2,0) circle (2.5pt);
    \draw[fill=Orchid] (-1,0) circle (2.5pt);
    \draw[fill=Orchid] (0,0) circle (2.5pt);
    \draw[fill=WildStrawberry] (1,0) circle (2.5pt);
    \draw[fill=RoyalBlue] (2,0) circle (2.5pt);
    \draw[fill=black] (3,0) circle (2.5pt);
    \draw[fill=Orchid] (4,0) circle (2.5pt);
    \draw[fill=Orchid] (5,0) circle (2.5pt);
    \draw[fill=RoyalBlue] (6,0) circle (2.5pt);
    
    \node at (-1,0.8) {\textcolor{RoyalBlue}{$\beta_1$}};
    \node at (2,0.8) {\textcolor{RoyalBlue}{$\beta_2$}};
    \node at (2,-0.8) {\textcolor{WildStrawberry}{$\alpha_1$}};
    \node at (-1,-0.6) {$\sigma_{k_1}$};
    \node at (0,-0.6) {$\sigma_{k_2}$};
    \node at (1,-0.6) {\textcolor{WildStrawberry}{$\sigma_{k_3}$}};
    \node at (3,-0.6) {\textcolor{black}{$\sigma_{k_4}$}};
    \node at (4,-0.6) {$\sigma_{k_5}$};
    \node at (5,-0.6) {$\sigma_{k_6}$};
    
    \end{tikzpicture}
    \centering
    \caption{Illustration of Case a and Case b. Here 
    The dashed blue arcs are in $C$, the red arc is in $\Touches$, and the black arc is in $\NonTouches$. The circle is in $B_1$, and the purple nodes are in $B_2$.}\label{fig:case32-r1r2}
    \end{figure}

    \item[]
    \begin{minipage}{14cm}  \textbf{Case c:}
        $u\in S_2(H_I)$  and $u$ appears an odd number of times
        in $H_I$.
        Let $k_1<\cdots<k_{2q+1}$ be the occurrences of $u$ in $I$
          for some $q\ge 0$.
    \end{minipage}
        
    \item[]
    \begin{minipage}{14cm}
    \textbf{Case c$_1$:}
        If the number of $\beta\in \Touches\cup C $ with $\vertex_H(\beta)=u$ is at least $1$,
          then we can handle this case similarly as Case 2: For each $i\in [q]$ we add
          $k_{2i-1}$ and $k_{2i}$ to $R_2$ if the arc $(k_{2i-1},k_{2i})$ is in $\NonTouches$, and we always keep 
          $k_{2q+1}$ in $I\setminus R$.
        In this case, the number of these $2q+1$ indices that do not get added to $R_3$ is
        $$
      1+O\Big(\text{number of $\beta\in \Touches\cup C$ with $\vertex_H({\beta})=u$}\Big)
      $$
      which remains an $O(\cdot)$ of the same quantity given that the latter is at least $1$.
    \end{minipage}
        
      \item[] \begin{minipage}{14cm} 
      \textbf{Case c$_2$:} Consider the case when there is no $\beta\in \Touches\cup C $ with          
        $\vertex_H(\beta)=u$. 
    We start with the easier situation when
      there is no $k\in I$ such that $\sigma_k=u$ and 
      $k \in \interior(\alpha )$ for some $\alpha \in C $.
    In this case we  add all $k\in I$
      with $\sigma_k=u$ to $R_3$.
    Note that the $(\vertex_H(\alpha ),u)$-th entry of the improvement vector
      of every $\alpha \in C $ is $0$. So
      removing all occurrences of $u$ has no effect.
      \end{minipage}
 
   \item[] \begin{minipage}{14cm}\textbf{Case c$_3$:} We are left with the case when there is no $\beta\in \Touches\cup C$ with $\vertex_H(\beta)=u$ and at the same time, there is    an arc $\alpha\in C$ such that $k_i\in \interior(\alpha)$ for some $i$. 
  Combining these two assumptions we must have that 
  $k_i\in \interior(\alpha)$ for all $i\in [2q+1]$.
      Given that $\alpha$ is a \shortradius arc, we have that the radius of both $k_1$ and $k_{2q+1}$
        is at most $2\cdot \Upper{\group_{i^*}}$.
      On the other hand, given that $\leng(I)=2\ell$
        and $\leng(\alpha)\le \ell$,
        the radius of either $k_1$ or $k_{2q+1}$ is at least $\ell/2\ge 2^{j^*-2}$.
If this holds for $k_1$,
    we add $k_2,\ldots,k_{2q+1}$ to $R_2$ and keep $k_1$ in $I\setminus R$; otherwise we add $k_1,\ldots,k_{2q}$ to
    $R_2$ and keep $k_{2q+1}$ in $I\setminus R$.
    In both cases the index left in $I\setminus R$ lies in $P$.
    \end{minipage}
\end{enumerate}
\end{flushleft}

  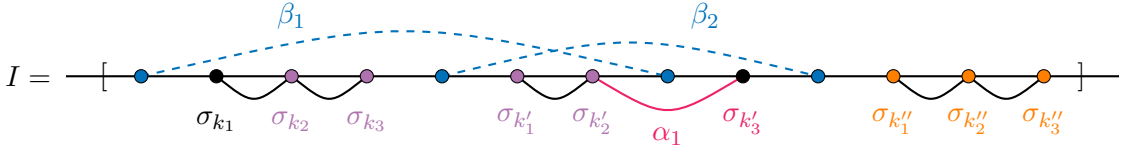
\begin{figure}[H]
    \begin{tikzpicture}
    \draw[thick] (-1,0) -- (13,0);
    \draw[thick,draw=RoyalBlue, dashed] (0,0) .. controls(3.5,0.8) .. (7,0);
    \draw[thick,draw=RoyalBlue, dashed] (4,0) .. controls(6.5,0.6) .. (9,0);
    \draw[thick,draw=WildStrawberry] (6,0) .. controls(7,-0.6) .. (8,0);
    \draw[thick] (1,0) .. controls(1.5,-0.4) .. (2,0);
    \draw[thick] (2,0) .. controls(2.5,-0.4) .. (3,0);
    \draw[thick] (5,0) .. controls(5.5,-0.4) .. (6,0);
    \draw[thick] (11,0) .. controls(11.5,-0.4) .. (12,0);
    \draw[thick] (10,0) .. controls(10.5,-0.4) .. (11,0);
    
    \draw[fill=RoyalBlue] (0,0) circle (2.5pt);
    \draw[fill=black] (1,0) circle (2.5pt);
    \draw[fill=Orchid] (2,0) circle (2.5pt);
    \draw[fill=Orchid] (3,0) circle (2.5pt);
    \draw[fill=RoyalBlue] (4,0) circle (2.5pt);
    \draw[fill=Orchid] (5,0) circle (2.5pt);
    \draw[fill=Orchid] (6,0) circle (2.5pt);
    \draw[fill=RoyalBlue] (7,0) circle (2.5pt);
    \draw[fill=black] (8,0) circle (2.5pt);
    \draw[fill=RoyalBlue] (9,0) circle (2.5pt);
    \draw[fill=orange] (10,0) circle (2.5pt);
    \draw[fill=orange] (11,0) circle (2.5pt);
    \draw[fill=orange] (12,0) circle (2.5pt);
    
    \node at (-0.5,0) {\textcolor{black}{$[$}};
    \node at (12.5,0) {\textcolor{black}{$]$}};
    \node at (-1.5,0) {\textcolor{black}{$I=$}};
    \node at (2,0.8) {\textcolor{RoyalBlue}{$\beta_1$}};
    \node at (7.5,0.8) {\textcolor{RoyalBlue}{$\beta_2$}};
    \node at (7,-0.8) {\textcolor{WildStrawberry}{$\alpha_1$}};
    \node at (1,-0.6) {\textcolor{black}{$\sigma_{k_1}$}};
    \node at (2,-0.6) {\textcolor{Orchid}{$\sigma_{k_2}$}};
    \node at (3,-0.6) {\textcolor{Orchid}{$\sigma_{k_3}$}};
    \node at (5,-0.6) {\textcolor{Orchid}{$\sigma_{k'_1}$}};
    \node at (6,-0.6) {\textcolor{Orchid}{$\sigma_{k'_2}$}};
    \node at (8,-0.6) {\textcolor{WildStrawberry}{$\sigma_{k'_3}$}};
    \node at (10,-0.6) {\textcolor{orange}{$\sigma_{k''_1}$}};
    \node at (11,-0.6) {\textcolor{orange}{$\sigma_{k''_2}$}};
    \node at (12,-0.6) {\textcolor{orange}{$\sigma_{k''_3}$}};
    
    \end{tikzpicture}
    \centering
    \caption{Illustration of Case c. Here $\vertex(\beta_1)$ and $\vertex(\beta_2)$ are adjacent to $\vertex(\alpha_1)$ in $G$.
    The dashed blue arcs are in $C$, the red arc is in $\Touches$, and the black arc is in $\NonTouches$. The orange nodes are in $B_3$, and the purple nodes are in $B_2$.}\label{fig:case32-r2r3}
    \end{figure}

Summarizing Case b and Case c, we have that $R_2$ consists of endpoints of a 
        collection of endpoint-disjoint arcs in $\NonTouches$.
      Moreover, the number of indices left in $I\setminus R $
        can be bounded by
\begin{equation}\label{heheha} 
        |P\cap I|+O (|\Touches\cup C | ).
\end{equation}

This gives us the following bound on $|I\setminus R|$:

\begin{lemma}\label{lem:largeratio}
 $|I\setminus R|\le O(\sqrt{\log n})\cdot |C |$.
\end{lemma}
\begin{proof}
This follows by combining (\ref{heheha}), (\ref{hahe1}) and (\ref{hehe100}).
\end{proof}

Finally we show that there is no change in the improvement
  vectors of $\alpha\in C$ after removing $R$:

\begin{lemma}\label{lem:nochange}
For every arc $\alpha\in C$ of $H$, 
  its corresponding arc $\beta=\rho(\alpha)\in Q$ of $B$ satisfies
$$
\improvement{\tau,B}{\beta}=\improvement{\gamma,H}{\alpha}.
$$
\end{lemma}
\begin{proof}
Let $\vertex_H(\alpha)=\vertex_B(\beta)=v$ and let $u$ be a vertex that 
  is adjacent to $v$ in $G$.
We consider two cases: the $(u,v)$-entry of $\improvement{\gamma,H}{\alpha }$ is $0$, and the $(u,v)$-entry is nonzero (either $-2$ or $2$).
For both cases our goal is to show that the $(u,v)$-entry of $\improvement{\tau,B}{\beta}$
  is the same.

For the first case, if $u$ does not appear inside $\alpha $ in $H$ then
  $u$ also does not appear inside $\beta$ in $B$.
Otherwise, $u$ appears $2q$ times inside $\alpha $ for some $q\ge 1$ and is covered in either Case b or Case c.
If it is in Case c$_2$, then all occurrences of $u$ are deleted 
  in $R_3$ and thus, the $(u,v)$-entry of $\improvement{\tau,B}{\beta}$ must be $0$.
Otherwise, a number of endpoint-disjoint arcs with node $u$
  in $\NonTouches$ are added to $R_2$ and deleted.
It follows that $u$ still appears an even number of times inside
  $\beta$ in $B$ and thus, the $(u,v)$-entry of $\improvement{\tau,B}{\beta}$
  remains $0$.
  
For the second case, $u$ appears an odd number of times inside
  $\alpha $ in $H$. 
If $u\in S_1(H_I)$, then~this unique appearance of $u$
  inside $\alpha$ was not deleted in Case 1 since 
  $\alpha$ is a \shortradius arc in $\group_{i^*}$. 
It remains in $I\setminus R$ and $u$ appears
  exactly once inside $\beta$ in $B$.
For the case when $u\in S_2(H_I)$,  
  it follows from a similar argument as in the first case 
   that $u$ still appears an odd number of times inside $\beta$
  in $B$.
\clearpage
We have shown that both $(u,v)$-entries in $\improvement{\gamma,H}{\alpha}$ and $\improvement{\tau,B}{\beta}$
  are nonzero; we now show that they have the same value.
Let $i$ be the first index of $I$. 
Recall that $\gamma'$ is the configuration of~$S(H)$ obtained after making the first 
  $i-1$ moves of $H$ on $\gamma$,  and that $\tau$ is the restriction of $\gamma'$ on $S(B)$.
To finish the proof, it suffices to show that the parity of the 
  number of occurrences of $u$ (and $v$) in $H_{[i:\leftt(\alpha )-1]}$
  is the same as that of the number of occurrences of $u$ ($v$, respectively)
  in $B_{[1:\leftt(\beta)-1]}$.
This is trivial for $v$ since $R_1\cup R_3$ does not contain any $k$ with $\sigma_k=v$ and the removal of arcs that are 
  endpoint-disjoint from $C$
  does not change the parity of this number.
For the case of $u$, if it appears only once in $I$ (inside $\alpha$), then 
  clearly the number we care about is $0$ in both~cases.
If $u\in S_2(H_I)$, then it is covered by Case b or c but not
  Case 3.2.
In both cases we delete a number of arcs with node $u$ that are
  in $\NonTouches$ and pairwise endpoint-disjoint.
It follows that the parity of this number does not change.
This finishes the proof of the lemma.
\end{proof}
\begin{lemma}[\answer{Lemma~\ref{lem:main} for Case 3.2}]\label{secondlowerbound:3.2}
Assume that condition (\ref{hehe100}) holds, i.e., $\Touches$ is small. Then there exist (i) a sequence $B$ of moves of length at most $5n$,
  (ii) a configuration $\tau$ of $S(B)$, and (iii) a set of arcs $Q$ of $B$ such that
\begin{enumerate}
    \item The rank of $Q$ in $B$ satisfies $\dfrac{\rank_B(Q)}{\leng(B)}\ge \Omega\left(\dfrac{1}{\sqrt{\log n}}\right) \quad\ \  (\text{A});$
    \item For every arc $\alpha\in Q$, there exists an arc $\alpha'$ of $H$ such that $\improvement{\tau,B}{\alpha}=\improvement{\gamma,H}{\alpha'}\quad\ \  (\text{B}).$
\end{enumerate}
\end{lemma}
\begin{proof}
Indeed using the interval $I\subseteq [m]$ of Equation~(\ref{hahe1}), we set $B=H_{I\setminus R}$ and $\tau$ be the restriction of configuration $\gamma'$ on $S(B)$.
We set also $Q$ the arcs of $B$ which are induced by collection $C=\restrict{C^*}{I}$ of (\ref{hahe1}). Lemma~\ref{lem:nochange} shows that (B)---Vector-Preservation--- property 
holds for $B,Q,\tau$. Finally the aforementioned analysis shows that (i) $\rank(C)$ is almost full or equivalently that $\rank_B(Q)\ge \Omega(|Q|)$ and (ii) using
Lemma~\ref{lem:largeratio} $\leng(B)=\leng(H_{I\setminus R})\le O(\sqrt{\log n})\cdot |Q|$ implying that (A)---High-rank property--- holds too.
\end{proof}

\section{Binary Max-CSP and Function problems}\label{sec:maxcsp}

\begin{definition}
An instance of Max-CSP (Constraint Satisfaction Problem) consists
of a set $V = \{ x_1, \ldots, x_n \}$ of variables that can take values over
a domain $D$, and a set $C = \{c_1, \ldots, c_m \}$ of constraints
with given respective weights $w_1, \ldots, w_m$.
A constraint $c_i$ is a pair $(R_i, t_i)$ consisting of a relation
$R_i$ over $D$ of some arity $r_i$ (i.e. $R_i \subseteq D^{r_i}$), 
and an $r_i$-tuple of variables (i.e., $t_i \in V^{r_i}$).
An assignment $\tau: V \rightarrow D$ satisfies the constraint $c_i$
if $\tau(t_i) \in R_i$.
The MAX CSP problem is: given an instance, find an assignment that
maximizes the sum of the weights of the satisfied constraints.
\end{definition}
We will focus here on the case of binary domains $D$, which wlog we can take to be
$\{0,1\}$, and binary relations ($r_i=2$); we refer to this as Binary Max-CSP,
or Max-2CSP. Several problems can be viewed as special cases of Binary Max-CSP where the
relations of the constraints are restricted to belong to
a fixed family ${\cal R}$ of relations; 
this restricted version is denoted Max-CSP(${\cal R}$).
For example, the Max Cut problem in graphs is equivalent to Max-CSP(${\cal R}$) where
${\cal R}$ contains only the ``not-equal'' (binary) relation $\neq$ (i.e.,
the relation $\{ (0,1), (1,0) \}$).
Other examples include:
\begin{itemize}
    \item {\em Directed Max Cut}. Given a directed graph with weights on its edges, partition the set of nodes into two sets
$V_0, V_1$ to maximize the weight of the edges that are directed from $V_0$ to $V_1$.
This problem is equivalent to Max-CSP(${\cal R}$) where ${\cal R}$ consists of the relation $\{(0,1)\}$.
 \item {\em Max 2SAT}. Given a weighted set of clauses with two literals in each clause, find a truth assignment that
maximizes the weight of the satisfied clauses.
This is equivalent to Max-CSP(${\cal R}$),
where ${\cal R}$ contains 4 relations, one for each of the 4 possible clauses with two literals
$a \lor b, {\bar a} \lor b, a \lor {\bar b}, {\bar a} \lor {\bar b}$; the relation for a clause contains the three assignments that satisfy the clause.
If we allow unary clauses in the 2SAT instance, then we include in ${\cal R}$ also the two unary constraints for $a$ and $\lnot a$. 
\item {\em Stable Neural Network}. A neural network in the Hopfield model \cite{hopfield} is an undirected graph $G=(V,E)$ (the nodes correspond to the neurons, the edges to the synapses) with a given weight $w_e$ for each edge $e \in E$ and a given threshold $t_u$ for each node $u \in V$; the weights and thresholds are not restricted in sign. A configuration $\gamma$ is an assignment of a value (its `state') $-1$ or $1$ to each node. A node $u$ is stable in a configuration $\gamma$ if $\gamma(u) = -1$ and $t_u + \sum_{v: (u,v) \in E} w_{(u,v)} \gamma_v \leq 0$, or
$\gamma(u) = 1$ and $t_u + \sum_{v: (u,v) \in E} w_{(u,v)} \gamma_v \geq 0$. A configuration is {\em stable} if all the nodes are stable in it. Hopfield showed that every neural network has one or more stable configurations, using a potential function argument: a node $u$ is unstable in a configuration $\gamma$ iff flipping its state increases the value of the potential function $p(\gamma) = \sum_{u \in V} t_u \cdot \gamma(u) + \sum_{(u,v) \in E} w_{(u,v)} \cdot \gamma(u) \gamma(v)$. Hence, $\gamma$ is stable iff $p(\gamma)$ cannot be increased by flipping the state of any node. Thus, the problem of finding a stable configuration is the same as the local  Max-CSP(${\cal R}$) problem when  ${\cal R}$ contains the unary constraint $a$ and the binary constraint $a = b$. The natural greedy algorithm, in which unstable nodes asynchronously (one-at-a-time) flip their state (in any order) monotonically increases the potential function and  converges to a stable configuration\footnote{Note that if unstable nodes flip their state simultaneously then the algorithm may oscillate and not converge to a stable configuration.}.
Our results apply to the smoothed analysis of this natural dynamics for neural networks.
\item {\em Network coordination game} with 2 strategies per player. We are given a graph $G=(V,E)$ where each node corresponds to a player with 2 strategies,
and each edge $(u,v)$ corresponds to a game $\Gamma_{u,v}$ between players $u ,v$ with a given
payoff matrix (both players get the same payoff in all cases). The total payoff of
each player for any strategy profile is the sum of the payoffs from
all the edges of the player. The problem is to find a pure equilibrium (there always exists one as these are potential games).
This problem can be viewed as a special case of local Max-CSP(${\cal R}$) 
where ${\cal R}$ contains the 4 singleton relations 
$\{(0,0)\}, \{(0,1)\}, \{(1,0)\}, \{(1,1)\}$,
and we want to find a locally optimal assignment that cannot be improved by flipping any single variable.
The FLIP algorithm in this case is the better response dynamics of the game.
Boodaghians et. al. \cite{coordination} studied the smoothed complexity of the better response algorithm for network coordination games, where each entry of every payoff matrix is independently drawn from a probability distribution supported on $[-1,1]$ with density at most $\phi$. They showed that for general graphs and $k$ strategies per player the complexity is at most
$\phi \cdot (nk)^{O(k \log (nk)}$ with probability $1- o(1)$, and in the case of complete graphs it is polynomial.
\end{itemize}

A constraint can be viewed as a function that maps each assignment for the
variables of the constraint to a value 1 or 0, depending on whether the
assignment satisfies the constraint or not.
We can consider more generally the {\em Binary function optimization problem} (BFOP),
where instead of constraints we have functions (of two arguments) with
more general range than $\{0,1\}$, for example $\{0,1,\ldots, k\}$,
for some $k$ fixed (or even polynomially bounded):
Given a set $V = \{ x_1, \ldots, x_n \}$ of variables with domain $D=\{0,1\}$, 
a set $F = \{f_1, \ldots, f_m \}$ of functions, where each
$f_i$ is a function of a pair $t_i$ of variables,
and given respective weights $w_1, \ldots, w_m$, find an assignment
$\tau: V \rightarrow D$ to the variables 
that maximizes $\sum_{i=1}^m w_i \cdot f_i (\tau(t_i))$.
In the local search version, we want to find an assignment that
cannot be improved by flipping the value of any variable.
For smoothed analysis, the weights $w_i$ are drawn independently
from given bounded distributions as in Max Cut.
We will show that the bounds for Max Cut extend to
the general Binary Max-CSP and Function optimization problems with arbitrary
(binary) constraints and functions.

Consider an instance of BFOP.  Even though a function (or a constraint in Max-2CSP) 
has two arguments, its value may depend on only one of them, i.e. it may be
essentially a unary function (or constraint).  More generally,
it may be the case that the function depends on both variables but the two variables can be decoupled and
the function can be separated into two unary functions.
We say that a binary function $f(x,y)$ is {\em separable}
if there are unary functions $f', f''$ such that
$f(x,y) = f'(x)+f''(y)$ for all values of $x,y$; otherwise $f$ is {\em nonseparable}.
For binary domains there is a simple easy criterion. 

\begin{lemma}
A binary function $f$ of two arguments with domain $\{0,1\}$ is separable
iff $f(0,0)+f(1,1) = f(0,1) +f(1,0)$.
\end{lemma}
\begin{proof}
If $f$ is separable, i.e., $f(x,y) = f'(x)+f''(y)$,
then $f(0,0)+f(1,1) = f'(0)+f''(0)+f'(1)+f''(1)=f(0,1) +f(1,0)$.
On the other hand, if $f(0,0)+f(1,1) = f(0,1) +f(1,0)$,
then we can define $f'(x) = f(0,0) + c \cdot x$,
where $c= f(1,0)-f(0,0) = f(1,1) - f(0,1)$,
and $f''(y)= d \cdot y$ where $d= f(0,1)-f(0,0) = f(1,1)-f(1,0)$.
It is easy to check that $f(x,y) = f'(x)+f''(y)$ for all values of $x, y \in \{0,1\}$.
\end{proof}

If our given instance of BFOP has some separable functions, then we can
replace them with the equivalent unary functions.
After this, we have a set of unary and binary functions, where
all the binary functions are nonseparable.

Consider a sequence $H$ of variable flips starting from an initial 
assignment (configuration) $\gamma \in \{0,1\}^n$.
When we flip a variable $x_j$ in some step, the change in the value
of the objective function can be expressed as $\langle \delta , X \rangle$,
where the coordinates of the vectors  $\delta , X$ correspond
to the functions of the given instance, the vector $\delta$ gives the
changes in the function values and $X$ is the vector of random weights of the functions.
Define the matrix $M_{H, \gamma}$ which has a row corresponding to each
nonseparable function $f_i$ and a column for each pair
of closest flips of the same variable in the sequence $H$, 
where the column is the sum of the vectors $\delta_1, \delta_2$ for the
two flips, restricted to the nonseparable functions.
Note that for separable functions, the corresponding 
coordinate of $\delta_1 + \delta_2 =0$.
Thus, the sum of the changes in the value of the objective function
in the two closest flips of the same variable is equal to the
inner product of the column $\delta$ with the vector of random weights of the nonseparable functions.

\begin{lemma}
The entry of the matrix $M_{H, \gamma}$ at the row for the (nonseparable) function $f_i$
and the column corresponding to an arc $\alpha$ of a variable $x_j$
is nonzero iff $x_j$ is one of the variables of the function $f_i$
and the other variable $x_k$ of $f_i$ appears an odd number of times
in the interior of $\alpha$.
\end{lemma}
\begin{proof}
If $x_j$ is not one of the variables of the function $f_i$ then
the entry is clearly 0. So suppose $x_j$ is one of the variables of $f_i$.
If the other variable $x_k$ appears an even number of times in the interior
of $\alpha$ then its value at the two flips of $x_j$ is the same
and the entry is again 0, regardless of what the function $f_i$ is.
So assume that $x_k$ appears an odd number of times in the interior of $\alpha$.
Then it is easy to check that the entry is equal to
$f_i(0,0)+f_i(1,1) - [f_i(0,1)+f_i(1,0)]$ if the variables $x_j, x_k$ have
different values before the first flip, and it has the opposite value
$f_i(0,1)+f_i(1,0) - [f_i(0,0)+f_i(1,1) ]$ if $x_j, x_k$ have the same value before the
first flip; in either case the entry is nonzero because $f_i$ is nonseparable.
\end{proof}

Thus, the zero-nonzero structure of the matrix $M_{H, \gamma}$ is the same
as that of the matrix for the Max Cut problem on the graph $G$ which has
the variables as nodes and has edges corresponding to the nonseparable functions
with respect to the same initial configuration $\gamma$ and sequence of flips $H$.
The arguments in the proof for the Max Cut problem that choose a subsequence
and bound the rank of the corresponding submatrix depend only on the 
zero-nonzero structure of the matrix and not on the precise values:
In every case, we identify a diagonal submatrix or a triangular submatrix of the appropriate size. 
Therefore, we can apply the same analysis for the general Max-2CSP and BFOP problems with arbitrary binary constraints or functions, proving Theorem \ref{maxcsptheorem}.


\section{Conclusions}\label{sec:conclusions}

We analyzed the smoothed complexity of the FLIP algorithm for local Max-Cut, and more generally, for binary maximum constraint satisfaction problems (like Max-2SAT, Max-Directed Cut, Stable Neural Network etc.). We showed that with high probability, every execution of the FLIP algorithm for these problems, under any pivoting rule, takes at most $\phi n^{O(\sqrt{\log n})}$ steps to terminate. The proof techniques involve a sophisticated analysis of the execution sequences of flips that are potentially generated by the FLIP algorithm, with the goal of identifying suitable subsequences (including non-contiguous subsequences) that contain many steps with linearly independent improvement vectors, which are preserved from the full execution sequence. We do not know at this~point~whether~the $\sqrt{\log n}$ in the exponent, which is due to the ratio between the length and the rank of the subsequence, can be improved or is best possible for this approach, i.e. whether our new rank lemma for subsequences is tight. 

There are several other interesting open questions raised by this work. One question concerns the extension to non-binary constraints. For example, does a similar result hold for Local Max-3SAT?  Does it hold generally for all Max-CSP problems with binary domains? There are several new challenges in addressing these questions. 

Another question concerns the extension to  domains of higher cardinality $k$. Simple examples of Max-2CSP with larger domain include Max-k-Cut, where the nodes are partitioned into $k$ parts instead of 2 as in the usual Max-Cut, and the Network Coordination Game with $k$ strategies per player. Bibak et. al. studied Max-k-Cut and showed that the FLIP algorithm converges with high probability in $\phi n^{O({\log n})}$ steps for general graphs for fixed $k$ (and polynomial time for complete graphs if $k=3$) \cite{bibak2019improving}. Boodaghians et. al. studied the network coordination game and showed a similar bound $\phi n^{O({\log n})}$ for general graphs for fixed $k$ (and in the case of complete graphs, polynomial time for all fixed $k$) \cite{coordination}.
Can the $\log n$ in the exponent be improved to $\sqrt{\log n}$ for these problems using a combination of the techniques in these papers and the present paper, and more generally does it hold for all Max-2CSP problems with non-binary domains?

Ultimately, is the true smoothed complexity of  Local Max-CSP problems polynomial or are there bad examples of instances and distributions that force super-polynomial behavior?

\clearpage
\appendix
\section{Proof of Lemma \ref{lem:arc_node_ratio_in_cover}}\label{sec:app1}

Let $\ell$ be a positive integer that satisfies $\ell\le m/2$.
To prove Lemma \ref{lem:arc_node_ratio_in_cover}, we define 
  $\Cover{\ell}$, a set of intervals of length $2\ell$ that aim to \emph{evenly}
  cover indices of $[m]$.
Formally, we have 
    %
    \[
    \Cover{\ell} = \Covereven{\ell}  \cup \Coverodd{\ell}  \cup \Coverboundary{\ell}  ~,
    \]
    where
    \begin{align*}
     &\Covereven{\ell} = \left\{ \big[(2i-2)\ell+1:2i\ell\big]: i=1,\ldots, \left\lfloor\frac{m}{2\ell}\right\rfloor \right\},
     \\[0.5ex]
     &\Coverodd{\ell} = \left\{\big[(2i-1)\ell+1:(2i+1)\ell\big] :i=1,\ldots, \left\lfloor\frac{m-\ell}{2\ell}\right\rfloor\right\} ,
     \\[0.6ex]
     &\Coverboundary{\ell} = \left\{\big[ m-2\ell+1:m\big]\right\}.
     \end{align*}
The following lemma summarizes properties we need about $\Cover{\ell}$.

    %

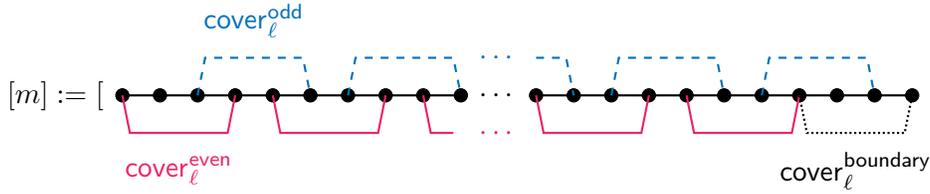
\begin{figure}[H]
\begin{tikzpicture}
\draw[fill=black] (-5,0) circle (2.5pt);
\draw[fill=black] (-4.5,0) circle (2.5pt);
\draw[fill=black] (-4,0) circle (2.5pt);
\draw[fill=black] (-3.5,0) circle (2.5pt);
\draw[fill=black] (-3,0) circle (2.5pt);
\draw[fill=black] (-2.5,0) circle (2.5pt);
\draw[fill=black] (-2,0) circle (2.5pt);
\draw[fill=black] (-1.5,0) circle (2.5pt);
\draw[fill=black] (-1,0) circle (2.5pt);
\draw[fill=black] (-0.5,0) circle (2.5pt);
\draw[fill=black] (5.5,0) circle (2.5pt);
\draw[fill=black] (5,0) circle (2.5pt);
\draw[fill=black] (4.5,0) circle (2.5pt);
\draw[fill=black] (4,0) circle (2.5pt);
\draw[fill=black] (3.5,0) circle (2.5pt);
\draw[fill=black] (3,0) circle (2.5pt);
\draw[fill=black] (2.5,0) circle (2.5pt);
\draw[fill=black] (2,0) circle (2.5pt);
\draw[fill=black] (1.5,0) circle (2.5pt);
\draw[fill=black] (1,0) circle (2.5pt);
\draw[fill=black] (0.5,0) circle (2.5pt);

\node at (-6,0) {$[m]:=$};
\node at (-5.3,0) {$[$};
\node at (0,0) {$\cdots$};
\node at (5.8,0) {$]$};
\node at (0,-0.5) {\textcolor{WildStrawberry}{$\cdots$}};
\node at (0,0.5) {\textcolor{RoyalBlue}{$\cdots$}};
\node at (-4.25,-1) {\textcolor{WildStrawberry}{$\Covereven{\ell} $}};
\node at (-3.25,1) {\textcolor{RoyalBlue}{$\Coverodd{\ell} $}};
\node at (4.75,-1)
{\textcolor{black}{$\Coverboundary{\ell} $}};

\draw[thick] (-5,0) -- (-0.5,0);
\draw[thick] (0.5,0) -- (5.5,0);
\draw[WildStrawberry, thick] (-5,0) -- (-4.9,-0.5) -- (-3.6,-0.5) -- (-3.5,0);
\draw[WildStrawberry, thick] (-3,0) -- (-2.9,-0.5) -- (-1.6,-0.5) -- (-1.5,0);
\draw[WildStrawberry, thick] (-1,0) -- (-0.9,-0.5) -- (-0.6,-0.5);
\draw[WildStrawberry, thick] (0.5,0) -- (0.6,-0.5) -- (1.9,-0.5) -- (2,0);
\draw[WildStrawberry, thick] (2.5,0) -- (2.6,-0.5) -- (3.9,-0.5) -- (4,0);
\draw[black, thick, densely dotted] (4,0) -- (4.1,-0.5) -- (5.4,-0.5) -- (5.5,0);
\draw[RoyalBlue, thick, dashed] (-4,0) -- (-3.9,0.5) -- (-2.6,0.5) -- (-2.5,0);
\draw[RoyalBlue, thick, dashed] (-2,0) -- (-1.9,0.5) -- (-0.6,0.5) -- (-0.5,0);
\draw[RoyalBlue, thick, dashed] (0.5,0.5) -- (0.9,0.5) -- (1,0);
\draw[RoyalBlue, thick, dashed] (1.5,0) -- (1.6,0.5) -- (2.9,0.5) -- (3,0);
\draw[RoyalBlue, thick, dashed] (3.5,0) -- (3.6,0.5) -- (4.9,0.5) -- (5,0);
\end{tikzpicture}
\centering
\caption{An illustration of $\Cover{\ell}$ with $\ell=2$.}
\end{figure}


\begin{lemma}\label{lem:arc_in_cover}
Each $i\in [m]$ is contained in at most three intervals in $\Cover{\ell}$.
Each arc $\alpha$ of $H$ with $\leng (\alpha) \le \ell$ is contained in at least
  one interval in $\Cover{\ell}$. 
\end{lemma}
\begin{proof}
The first part of the lemma is trivial. For the second part, let  $\alpha=(i,j)=(\leftt(\alpha),\rightt(\alpha))$ be
an arbitrary arc of length at most $\ell$. Let $k= \lfloor \frac{m}{\ell} \rfloor \ge 2$ (since $m\ge 2\ell$) be the largest integer such that $k\ell\le m$.
Notice that if $m$ is a multiple of $\ell$, then $\Coverboundary{\ell}$ is included in either 
$\Covereven{\ell}$ or $\Coverodd{\ell}$ depending on the parity of $\left \lfloor \frac{m}{\ell} \right \rfloor$. 
We will split into two cases, depending on whether $\rightt(\alpha) \le k\ell= \left \lfloor \frac{m}{\ell} \right \rfloor\ell]$ or  $\rightt(\alpha) > k\ell $.

     \textbf{Case 1:} If $\rightt(\alpha)\le k\ell=\left \lfloor \frac{m}{\ell} \right \rfloor\ell$, then we claim that $\alpha$ is contained in one of the intervals in
  $\Covereven{\ell}\cup \Coverodd{\ell}$.To see this, let $k'=\left \lceil \tfrac{\rightt(a)}{\ell} \right \rceil \le k$ be the smallest integer such that $\rightt(a)\le k'\ell$ (so $\rightt(\alpha)> (k'-1)\ell$).
If $\rightt(\alpha)\le \ell$ $(k'=1)$, $\alpha$ is trivially contained in $[ \ell]\subset [2\ell]$ and we are done;  so we assume below that $k'\ge 2$. Since $\leng(\alpha)\le \ell$ or equivalently $\rightt(\alpha)-\leftt(\alpha)+1\le \ell$ , it holds that $\leftt(\alpha)\ge \rightt(\alpha)-\ell+1>(k'-2)\ell +1$ and thus, $\alpha=(\leftt(\alpha),\rightt(\alpha))$ is contained in $[(k'-2)\ell+1:  k'\ell]$, 
one of the intervals in $\Covereven{\ell}\cup \Coverodd{\ell}$.

     \textbf{Case 2:} We are left with the case when $\rightt(\alpha)> k\ell=\left \lfloor \frac{m}{\ell} \right \rfloor\ell$. 
Using $m<(k+1)\ell$ (by the choice of $k$), we have that  $\leftt(\alpha)\ge \rightt(a)-\ell+1>(k-1)\ell+1>m-2\ell+1.$
As a result, $\alpha$ is contained in the interval in $\Coverboundary{\ell}$.

\end{proof}
We are now ready to use $\Cover{\ell}$ to prove Lemma \ref{lem:arc_node_ratio_in_cover}.

\begin{proof}[Proof of Lemma \ref{lem:arc_node_ratio_in_cover}]
We first partition $\Cover{\ell}$ into
\begin{align*}
    \mathcal{I}_\text{large} = \left\{I\in \Cover{\ell}:
    \frac{\left|\restrict {C}{I}\right|}{|C|} \ge \frac{\ell}{8m}\right\}\ \quad\text{and}\ \quad
    \mathcal{I}_\text{small} = \left\{I\in \Cover{\ell}:\frac{\left|\restrict{C}{I}\right|}{|C|} < \frac{\ell}{8m}\right\}.
\end{align*}
    We would like to show that there is an interval in $\mathcal{I}_\text{large}$
    that satisfies the second inequality in (\ref{hehe2}).
    
    From the definition of $\mathcal{I}_\text{small}$ we have
    \[
    \sum_{I\in \mathcal{I}_\text{small}} \frac{\left|\restrict{C}{I}\right|}{|C|}<
    \frac{\ell}{8m}\cdot \big|\Cover{\ell}\big|\le
    \frac{\ell}{8m}\cdot \left(2\cdot \left\lfloor \frac{m}{2\ell}\right\rfloor+1\right)\le \frac{1}{4}.
    \]
It follows from the second part of Lemma~\ref{lem:arc_in_cover} that each arc is contained in at least one interval of $\Cover{\ell}$ and thus $\sum_{I\in \Cover{\ell}} \left|\restrict{C}{I}\right|\ge |C|$ or equivalently
\begin{equation}\label{hehe3}
    \sum_{I\in \mathcal{I}_\text{large}} \frac{\left|\restrict{C}{I}\right|}{|C|} >  \frac{3}{4}.
\end{equation}
Assume for a contradiction that every interval $I \in \mathcal{I}_\text{large}$ satisfies that 
$\dfrac{\left| \restrict{C}{I} \right|}{|C|}< 
       \dfrac{|P\cap I|}{4|P|}.
$
Then  
$$
\frac{3}{4}<\sum_{I\in \mathcal{I}_\text{large}}\frac{\left| \restrict{C}{I} \right|}{|C|}< 
\sum_{I\in \mathcal{I}_\text{large}}
       \frac{|P\cap I|}{4|P|}.\vspace{0.08cm}
$$
However, it follows from the first part of Lemma \ref{lem:arc_in_cover}  that each $i\in P$ is contained in $I$ at most three times, and thus
$$
\sum_{I\in \Cover{\ell}} |P\cap I|\le 3|P|,
$$
which leads to a contradiction with (\ref{hehe3}).

\end{proof}

\section{A sequence that has low rank in every scale}\label{sec:hard-case}

Given a sequence $S$ of moves, we write $\rank(S)$ to denote 
  the rank of the set of all arcs in $S$.
In this section, we will construct a sequence $H$
   of length $5n$, specify the subgraph of 
  $G$ over active nodes in $H$,
  and show that every substring $S$ of $H$
  satisfies 
\[\frac{\rank (S)}{\leng(S)}= O\left(\frac{1}{\log n}\right).\]

For simplicity, let $d=0.1\log_3 n$. The set of active nodes
  in $H$ are $\bigcup_{k=1}^{d}V_k$, where $$V_k=\big\{v_{k,0},v_{k,2},\ldots, v_{k,n^{0.1}\cdot 3^{k-1}-1}\big\}.$$ 
Let $N_k=n^{0.1}\cdot 3^{k-1}$ denote the number of  nodes in $V_k$. So there are $O(n^{0.2})\ll n$ active nodes in $H$ in total. The subgraph of $G$ over these nodes contains all
  edges $(u,v)$ with $u\in V_1$ and $v\in \bigcup_{k=2}^d V_k$. 

Now we construct the sequence $H$ as the concatenation of  $5n/d$ blocks:
\[\mathcal{H}:= B_1\circ B_2 \circ \cdots \circ B_{5n/d}.\]
In each block $B_i$, there are $d$ moves, one from each $V_k$:
\begin{equation*}
    B_i = \Big(v_{1,i\hspace{-0.22cm}\pmod{N_1}} , v_{2,i\hspace{-0.22cm}\pmod{N_2}}, \cdots , v_{d, i\hspace{-0.22cm}\pmod{N_d}}\Big)
    ~.
\end{equation*}

We will prove the following bound on the rank of any substring $S$ of $H$:

\begin{lemma}\label{lem:hardness_logn_rank}
    For any substring $S$ of $H$ constructed above, we have
    \begin{equation*}
        \frac{\rank(S)}{\leng (S)} = O\left(\frac{1}{\log n}\right).
    \end{equation*}
\end{lemma}
\begin{proof}
Note that to have a repeating node in $S$,
  its length needs to be at least $N_1d +1$.
As a result we may assume without loss of generality that 
  $S$ is the concatenation of $t$ blocks for some 
  $t>N_1$.
Every node $v_{k,j}$ in $V_k$ 
  with $N_k+1>t$ appears at most once in $S$ and thus,
  it does not contribute anything to the rank of $S$.
Let $k^*$ denote the largest $k$ with $N_k+1\le t$.

For each $k:2\le k\le k^*$ and each $v_{k,j}\in V_k$, 
  we show that all arcs with node $v_{k,j}$ (if any) in $S$
  must share the same improvement vector and thus,
  contribute at most $1$ to the rank.
This follows from the facts that $v_{k,j}$ is only adjacent
  to nodes in $V_1$ and that every arc of $v_{k,j}$ contains $3^{k-1}$ (an odd number) occurrences of each node in $V_1$.
Finally, we have that the contribution to the rank of $S$ from
  each node $v_{1,j}\in V_1$ is trivially at most  its number
  of arcs in $S$.

Combining all these bounds, we are ready to bound the rank of $S$:
\begin{align*}
     \rank (S)\le \sum_{k=2}^{k^*} |V_k| + |V_1|\cdot \left\lceil\frac{t}{N_1}\right\rceil
     =O(N_{k^*})+O(t)=O(t)=O\left(\frac{\leng(S)}{\log n}\right),
\end{align*}
where in the first equation we used $t>N_1$ and in the 
  last equation we used $d=\Theta(\log n)$.
\end{proof}




\begin{flushleft}
\bibliographystyle{acm}
\bibliography{smooth-complexity-refs}
\end{flushleft}

\end{document}